%% file: main.tex
\documentclass[a4paper,USenglish,cleveref, autoref, thm-restate]{lipics-v2021}
\usepackage{mydefs}
\title{Who Wins the Multi-Structural Game?}

\author{Ronald {Fagin}}{IBM Research Almaden, 650 Harry Road, San Jose CA 95120, USA \and \url{https://research.ibm.com/people/ron-fagin} }{fagin@us.ibm.com}{}{}

\author{Neil {Immerman}}{University of Massachusetts Amherst, 140 Governors Drive, Amherst MA 01002, USA \and \url{https://people.cs.umass.edu/~immerman/} }{immerman@umass.edu}{https://orcid.org/0000-0001-6609-5952}{}

\author{Phokion {Kolaitis}}{IBM Research Almaden, 650 Harry Road, San Jose CA 95120, USA \and University of California Santa Cruz, 1156 High Street, Santa Cruz CA 95064, USA \and \url{https://users.soe.ucsc.edu/~kolaitis/} }{kolaitis@ibm.com}{https://orcid.org/0000-0002-8407-8563}{}

\author{Jonathan {Lenchner}}{IBM Research Yorktown Heights, 1101 Kitchawan Road, Yorktown Heights NY 10598, USA \and \url{https://research.ibm.com/people/jonathan-lenchner} }{lenchner@us.ibm.com}{https://orcid.org/0000-0002-9427-8470}{}

\author{Rik {Sengupta}}{IBM Research Cambridge, 314 Main Street, Cambridge MA 02142, USA \and \url{https://people.cs.umass.edu/~rsengupta/} }{rik@ibm.com}{https://orcid.org/0000-0002-9238-5408}{}

\authorrunning{R. Fagin, N. Immerman, P. Kolaitis, J. Lenchner and R. Sengupta} %TODO mandatory. First: Use abbreviated first/middle names. Second (only in severe cases): Use first author plus 'et al.'

\Copyright{Ronald Fagin, Neil Immerman, Phokion Kolaitis, Jonathan Lenchner and Rik Sengupta}




\category{}

\relatedversion{} %optional, e.g. full version hosted on arXiv, HAL, or other respository/website
\acknowledgements{We wish to thank Marco Carmosino for invaluable discussions during this work.}

\nolinenumbers %uncomment to disable line numbering

\EventEditors{John Q. Open and Joan R. Access}
\EventNoEds{2}
\EventLongTitle{42nd Conference on Very Important Topics (CVIT 2016)}
\EventShortTitle{CVIT 2016}
\EventAcronym{CVIT}
\EventYear{2016}
\EventDate{December 24--27, 2016}
\EventLocation{Little Whinging, United Kingdom}
\EventLogo{}
\SeriesVolume{42}
\ArticleNo{23}
\begin{document}

\maketitle

\begin{abstract}
Combinatorial games  played between two players,  called Spoiler and Duplicator,  have often been used to capture syntactic properties of formal logical languages. For instance, the widely used Ehrenfeucht-Fra\"{i}ss\'{e} (EF) game captures the syntactic measure of quantifier rank of first-order formulas. For every such game, there is an associated natural decision problem: ``given an instance of the game, does Spoiler win the game on that instance?'' For EF games, this problem was shown to be PSPACE-complete by Pezzoli in 1998. In this present paper, we show that the same problem for the \emph{multi-structural} (MS) games of recent interest is PSPACE-hard, but contained in NEXPTIME. In the process, we also resolve an open problem posed by Pezzoli about the dependence of the hardness results for EF games on the arity of the schema under consideration. Our techniques combine adaptations of Pezzoli's constructions together with insights from the theory of inapproximability of optimization problems, as well as the recently developed technique of parallel play for MS games.
\end{abstract}

\input{01intro}
\input{02prelims}
\input{03upper}
\input{04lower}
\input{05discussion}

\bibliography{bibliography}

\end{document}

%% file: 01intro.tex
\section{Introduction}
\label{sec:intro}

Combinatorial games are one of the main tools for studying the expressive power of logics on classes of structures. In particular, these games are indispensable in the study of logics on classes of finite structures, since the compactness theorem and other basic results of mathematical logic fail in the finite realm. Examples of such combinatorial games include the {\EF} game \cite{fraisse1950nouvelle,ehrenfeucht1961application}, the pebble game \cite{DBLP:journals/jsyml/Barwise77,DBLP:journals/jcss/Immerman82}, and the existential pebble game \cite{DBLP:journals/jcss/KolaitisV95}. The {\EF} game captures the quantifier depth of first-order formulas,
%\jon{Citation needed, e.g., either to the Libkin or Immerman book}
and thus it is used to analyze the expressive power of first-order logic. The pebble game captures the number of distinct variables in finite-variable infinitary logic, and thus it is used to analyze the expressive power of these logics,
%\jon{Citation needed}
%\phokion{No, the citations are already given in the previous sentence}
as well as the expressive power of least fixed-point logic since, on classes of finite structures,  the latter is a fragment of finite-variable infinitary logic. Finally, the existential pebble game captures the existential fragment of finite-variable infinitary logic,
%\jon{Citation needed}
and thus it can be used to analyze the expressive power of Datalog on finite structures.
%\jon{Probably should give a citation to some exposition of Datalog}
%\phokion{Again, the citations are already there, including for Datalog}
Detailed expositions of these games are in the books \cite{DBLP:series/txtcs/GradelKLMSVVW07,immermanbook,DBLP:books/sp/Libkin04}.

Each of these games is not a single game, but rather a family of games parameterized by a positive integer. Concretely, {\EF} games are parameterized by the number $m$ of moves in the game, while pebble games and existential pebble games are parameterized by the number $k$ of pebbles. Furthermore, each such parameterized game is played by two players, called Spoiler and Duplicator, on a pair $(\bA, \bB)$ of structures. For each of these games, it can be proved that Duplicator  wins the game at hand on $(\bA, \bB)$ if and only if the structures $\bA$ and $\bB$ cannot  be distinguished by a sentence of a certain logic (this is the logic that the game is said to \emph{capture}).
%\jon{This is a statement of what it means for a game to capture a syntactic property, not something additional, so I find this statement, as written, to be not quite right.}\rik{I think that's what's being spelled out here explicitly.}\jon{To me it is just the opposite. You are leaving the word `capture' undefined, but saying that one can prove that Duplicator wins ... by a sentence ... that the game captures. I think you should explicitly state the meaning of the word 'capture' and say that it is justified by the stated generic result.};
%\jon{I think these citations are coming too late.}

Each of the three games mentioned above gives rise to a natural algorithmic problem; namely, given an input to the game, determine the winner of the game on the given input. In fact, for each fixed schema $\tau$, there are two versions of this algorithmic problem, depending on whether or not the parameter associated with the game is part of the input. Specifically,  the {\EF} game gives rise to the following decision problem \winef: given two finite $\tau$-structures $\bA$ and $\bB$, and a positive integer $m$, does Spoiler win the $m$-move {\EF} game on $(\bA, \bB)$? Furthermore, for each fixed $m\geq 1$, the {\EF} game gives rise to the following decision problem $\winef_m$: given two $\tau$-structures $\bA$ and $\bB$, does Spoiler win the $m$-move {\EF} game on $(\bA, \bB)$?
%\jon{I do not understand the difference between $\winef$ and $\winef_m$. It seems that in both cases you are given $(\bA, \bB)$ and $m$. What am I missing? It also seems to me that the game should be denoted by $\textrm{WIN}(\bA, \bB)_m$ and that in all cases the size of the input is $\max(|\bA|, |\bB|,m)$. I believe the problem that is in LOGSPACE is that of deciding whether Spoiler has won given a particular play of the game (i.e., partial iso-testing).}
%\rik{I think I disagree on both counts. Notationally, $\winef_m$ should not be replaced with $\textrm{WIN}(\bA, \bB)_m$, because $\bA$ and $\bB$ are part of the input; for instance, we don't refer to the decision problem as \textsc{Clique}$(G)$, just as \textsc{Clique}. Also, $\winef$ has $m$ as part of the input, and $\winef_m$ does not (i.e. it takes $m$ as a universal constant). That's the critical difference. We are not interested separately in this paper in the partial isomorphism testing problem you mention at the end.}.
%\jon{OK. I think I understand and agree with you. $\winef_m$ is analogous to \textsc{Clique}$_k$($G$), where in the analysis you treat $k = O(1)$ and analogously, with $\winef_m$ you treat $m = O(1)$.}
In a similar manner, the pebble game gives rise to the decision problems {\winpb} and $\winpb_k$, for $k \geq 1$, while the existential pebble game gives rise to the decision problems $\winepb$ and $\winepb_k$, for $k\geq 1$. 
%It is easy to see that in the problems
%\winef, \winpb, and \winepb, it makes no difference whether the parameter of the game is given in unary or in binary notation.

The state of affairs concerning the computational complexity of these decision problems is as follows. For the {\EF} game,  the results in \cite{fraisse1950nouvelle} imply that $\winef_m$ is in $\LOGSPACE$, for each $m\geq 1$. In contrast, Pezzoli \cite{Pezzoli98} showed that $\winef$ is a $\PSPACE$-complete problem. For the pebble game, Grohe \cite{GroheLk} showed that $\winpb_k$ is a $\PTIME$-complete problem, for each $k\geq 2$. The problem {\winpb} is in $\EXP$, but its exact computational complexity remains open. Finally, for the existential pebble game, Kolaitis and Panttaja \cite{DBLP:conf/csl/KolaitisP03} showed that 
$\winepb_k$ is a $\PTIME$-complete problem for each $k\geq 2$, while $\winepb$ is an $\EXP$-complete problem.

In this paper, we focus on a different combinatorial game, called the \emph{multi-structural} (\ms) game. This game was introduced by Immerman \cite{MScanon0}, who showed that the {\ms} game captures the number of quantifiers in first-order formulas. There are two main differences between the {\EF} game and the {\ms} game: first, the {\ms} game is played on a pair $(\cA, \cB)$ of \emph{sets} of structures, instead of being played on a pair $(\bA, \bB)$ of structures; second, in every round of the {\ms} game, Duplicator can make an arbitrary number of copies of structures on the side that she plays in, effectively simulating all possible 
%\jon{EF-game-type}
responses simultaneously, and only has to exhibit her win on a single pair of structures at the end, one on either side.
%\jon{There is also the fact/difference that Duplicator just has to exhibit a partial isomorphism on a single pair of boards to win the game.}
Because of these differences, {\ms} games turn out to be significantly harder to analyze than {\ef} games. In fact, no substantive progress in the analysis of {\ms} games had been made prior to \cite{DBLP:conf/lics/FaginLR021}, which provided the impetus for further investigations of {\ms} games, including \cite{DBLP:conf/mfcs/FaginLVW22,DBLP:journals/lmcs/CarmosinoFIKLS24,parallelplay}. Furthermore, several different games that capture the number of quantifiers in first-order formulas (and hence, equivalent to {\ms} games) were studied in \cite{HL24}.

Here, we investigate the algorithmic problem of determining the winner of the {\ms} game. More formally, for each fixed schema $\tau$, we consider the decision problem \winms: given two finite sets $\cA$ and $\cB$ of finite $\tau$-structure, and a positive integer $m$, does Spoiler win the $m$-move {\ms} game on $(\cA, \cB)$?  In addition, for each fixed $m\geq 1$, we consider the decision problem $\winms_m$: given two finite sets $\cA$ and $\cB$ of finite $\tau$-structure, does Spoiler win the $m$-move {\ms} game on $(\cA, \cB)$?

We first observe that, for each fixed $m\geq 1$, the problem $\winms_m$ is in $\LOGSPACE$ (Proposition \ref{prop:winmsmupper}). We next explore the computational complexity of $\winms$, and establish an upper bound and a lower bound for this problem. For the upper bound, it is relatively straightforward to show that $\winms$ is in $\NEXP$, i.e., {\winms} is solvable in non-deterministic exponential time (Proposition \ref{prop:winmsupper}).

Establishing lower bounds for this problem, on the other hand, turns out to be a much more challenging task. Our main result (Theorem \ref{thm:winmslower}) asserts that {\winms} is a $\PSPACE$-hard problem, which implies that {\winms} is at least as hard as \winef.

Our approach involves combining an adaptation of the sophisticated gadgets from \cite{Pezzoli98} (which in turn were inspired by the gadgets introduced in Cai, F\"{u}rer, and Immerman \cite{CFI} to study the expressive power of first-order logic with counting) with two critical new ingredients: (i) the recently developed technique of \emph{parallel play} from \cite{parallelplay}, which enables us to simulate several parallel {\ms} games without wasting too many moves; and (ii) the inapproximability results for $\PSPACE$-hard problems in \cite{DBLP:journals/cjtcs/CondonFLS95}.
%Despite being three decades old, it appears that this last toolbox has not been adapted to this realm of study before this present work.
We exhibit a natural reduction from {\winms} to an optimization variant of \textsc{QBF} known as \textsc{Max-Q3SAT}, which is known to be $\PSPACE$-hard to approximate within a constant multiplicative factor. Furthermore, we show that our reduction holds even for directed graphs. In particular, along the way, we answer a question raised by Pezzoli in \cite{Pezzoli98} in the positive, namely, whether the $\PSPACE$-hardness of the {\EF} game holds for schemas in which each relation has arity at most two; Pezzoli \cite{Pezzoli98} established the $\PSPACE$-hardness of $\winef$ only for schemas containing a ternary relation symbol.

We leave open the problem of closing the gap between $\PSPACE$-hardness and containment in $\NEXP$. We note that our new techniques establishing the $\PSPACE$ lower bound do not carry over to $\NEXP$-hardness in a straightforward way. We conjecture that $\winms$ is a $\NEXP$-hard problem (and therefore $\NEXP$-complete as well). We hope that this conjecture will stimulate research in this fundamental algorithmic problem.

%% file: 02prelims.tex
\section{Preliminaries and Background}
\label{sec:prelims}

A \emph{schema} is a tuple
$\tau =(R_1,\ldots,R_t,c_1,\ldots,c_s)$ of relation symbols $R_i$   and constant symbols $c_j$, where each $R_i$ has a designated positive integer as its arity. A  \emph{$\tau$-structure} is a tuple ${\bf A}=(A,R_1^\bA,\ldots,R_t^\bA,c_1^\bA,\ldots,c_s^\bA)$, where $A$ is a set, called the \emph{universe} of $\bA$, each $R_i^\bA$ is a relation on $A$ whose arity matches that of the relation symbol $R_i$, and each $c_j^\bA$ is an element of  $A$. We say that a structure $\bA$ is \emph{finite} if its universe is a finite set.
In what follows, we will denote $\tau$-structures in boldface (e.g., $\bA$, $\bB$), their universes in capital letters (e.g., $A$, $B$), and sets of pebbled $\tau$-structures (defined next) in calligraphic typeface (e.g., $\cA$, $\cB$). We adopt the terminology and notation of recent papers on multi-structural games (e.g., \cite{parallelplay}).

\subsection{Pebbled Structures}

 We assume that we have a palette $\cC = \{x_1, x_2, \ldots\}$ of \emph{colors}, and that for each color, we have infinitely many pebbles of that color available. A $\tau$-structure $\bA$ is \emph{pebbled} if some elements  in the universe $A$ of $\bA$ have pebbles on them, provided there is at most one pebble of each color on an element. In general, an element in  $A$ may have multiple pebbles (of different colors) on it.
 
If pebbles of colors $x_1, \ldots, x_k$ are placed  on (not necessarily  distinct) elements $a_1, \ldots, a_k$ in the universe of  a  $\tau$-structure $\bA$, the resulting pebbled $\tau$-structure is called  \emph{$k$-pebbled}, and is denoted  by $\langle \bA ~|~ a_1, \ldots, a_k\rangle$. Note that $\bA$ can be viewed as the $0$-pebbled structure $\langle \bA ~|~\rangle$. If the number of pebbles is clear from the context, we 
write $\langle \bA ~|~ \vec{a}\rangle$ for a pebbled structure.

We remark at this juncture that the colors are designated $x_1, \ldots, x_k$ deliberately: colored pebbles will turn out to correspond to variables; thus, when a pebble colored $x_i$ is placed  on an element  $a_i \in A$, we will implicitly have that $x_i$ is interpreted as the element $a_i$ on $\bA$.

We say that two pebbled structures $\langle \bA ~|~ a_1, \ldots, a_k\rangle$ and $\langle \bB ~|~ b_1, \ldots, b_\ell\rangle$ are a \emph{matching pair} if $k = \ell$, and the function $\iota : A \to B$ defined by:

\centerline{$\iota(a_i) = b_i$, for  $1 \leq i \leq k$  and      $\iota(c^\bA) = c^\bB$, for every constant $c$ in $\tau$,}
\noindent is an isomorphism between the substructures of $\bA$ and $\bB$ induced by $a_1,\ldots,a_k$ and $b_1,\dots,b_k$.

\subsection{Multi-Structural Games}

An instance of the multi-structural game consists of a natural number $m$, and two finite sets $\cA$ and $\cB$ of  pebbled finite $\tau$-structures, each pebbled with the \emph{same} subset $\{x_1, \ldots, x_k\} \subseteq \cC$ of colors.
The \emph{$m$-round multi-structural (MS) game on $(\cA, \cB)$} is a game played by two players, called Spoiler ($\bS$, he/him) and  Duplicator ($\bD$, she/her), over $m$ rounds. In a run of the game, the $i$-th round, where  $1 \leq i \leq m$, consists of two steps:
\begin{itemize}
    \item First, $\bS$ chooses either $\cA$ (the ``left side'') or $\cB$ (the ``right side''), and an unused color $x_t \in \cC$; he then places (``plays'' or ``moves'') a pebble of color $x_t$ on an element of \emph{every} pebbled structure in the chosen side.
    \item Next, $\bD$ makes as many copies as she wants of each pebbled structure on the other side, and then places a pebble of color $x_t$ on an element of each such copy.
\end{itemize}
$\bD$ wins this run of the game if at the end of round $m$, there is a pebbled structure in $\cA$ and a pebbled structure in $\cB$ forming a matching pair. Otherwise, $\bS$ wins the run. We say $\bS$ has a \emph{winning strategy} on an instance $\langle m, (\cA, \cB)\rangle$ if he has a way to respond to the moves made by $\bD$ so that he wins the resulting run of the game, regardless of the moves made by $\bD$. 
We define a winning strategy for $\bD$ analogously.  It is easy to see that the game is \emph{determined}, i.e., exactly one of $\bS$ and $\bD$ has a winning strategy on a given instance. Therefore, every  instance of the game can be classified as either \emph{$\bS$-winnable} or \emph{$\bD$-winnable}.

We adopt the convention of always calling the two sides $\cA$ and $\cB$, even though they change over the course of a run of the game, since the Duplicator adds  new copies and since more pebbles are places on each pebbled structure.

For sets $\cA$ and $\cB$ described as above and with all structures in $\cA \cup \cB$ pebbled with the same subset $\{x_1, \ldots, x_k\} \subseteq \cC$ of colors, we say that a FO-formula $\varphi$ \emph{separates} $(\cA, \cB)$ if 
\begin{itemize}
    \item $\mathsf{free}(\varphi) = \{x_1, \ldots, x_k\}$   (i.e., $\{x_1,\ldots,x_k\}$ is the set of the free variables of $\varphi$);
    \item  $\bA \models \varphi[a_1/x_1, \ldots, a_k/x_k]$,   for every $\langle \bA ~|~ a_1, \ldots, a_k\rangle \in \cA$; 
    \item $\bB \models \lnot\varphi[b_1/x_1, \ldots, b_k/x_k]$, for  every $\langle\bB ~|~ b_1, \ldots, b_k\rangle \in \cB$.
\end{itemize}
We call $\varphi$ a \emph{separating sentence} for $(\cA, \cB)$ if $\cA$ and $\cB$ are sets of unpebbled structures.

The following key theorem \cite{MScanon0, DBLP:journals/lmcs/FaginLRV25} yields the precise sense in which multi-structural games capture the number of quantifiers of FO-formulas.

%relates the logical characterization of a separating formula with the combinatorial property of a game strategy.

\begin{theorem}[Fundamental Theorem of MS Games, \cite{MScanon0,DBLP:journals/lmcs/FaginLRV25}]\label{thm:MSfundamental}
Let  $\cA$ and $\cB$ be two finite sets of  pebbled finite $\tau$-structures, each pebbled with the same subset  of colors. Then the following statements are equivalent:
\begin{enumerate}
\item
 $\bS$ has a winning strategy in the $m$-round MS game on $(\cA, \cB)$.
 \item There is a separating formula for $(\cA, \cB)$ with at most $m$ quantifiers.
 \end{enumerate}
\end{theorem}

The {\ms} game can also be regarded as  a single-player game, since $\bD$ has a clear optimal \emph{oblivious} strategy: in each round, $\bD$ can make enough copies of each pebbled structure on her side in order to play all possible responses (one on each copy). In $\bD$-winnable instances, it can be easily checked that the oblivious strategy is a  winning strategy.

The following simple observation helps keep our proofs clean.

\begin{lemma}\label{lem:discard}
At any point during a run of an {\ms} game, if there is a pebbled structure $\langle\bA ~|~ \vec{a}\rangle$ not forming a matching pair with any pebbled structure on the other side, then $\langle\bA ~|~ \vec{a}\rangle$ can be discarded from the game with no effect on the final outcome of the game.
\end{lemma}

\subsection{Parallel Play}

 The technique of \emph{parallel play}, introduced in \cite{parallelplay}, allows the Spoiler  to partition an MS game into MS sub-games corresponding to distinct isomorphism classes (on the pebbles played so far), and then play different {\ms} games on those isomorphism classes in parallel, thereby potentially saving himself a large number of redundant moves that he might have needed with a na\"{i}ve strategy. The basic idea is straightforward: suppose Spoiler is faced with a game on $(\cA, \cB)$, where $\cA = \cA_1 \sqcup \cA_2$,   $\cB = \cB_1 \sqcup \cB_2$, and the followin hold:
\begin{itemize}
    \item no pebbled structure in $\cA_1$ forms a matching pair with a pebbled structure in $\cB_2$, and
    \item no pebbled structure in $\cA_2$ forms a matching pair with a pebbled structure in $\cB_1$.
\end{itemize}
Suppose further that Spoiler can win the $r$-round {\ms} game on $(\cA_1, \cB_1)$, and he can win the $r$-round {\ms} game on $(\cA_2, \cB_2)$, \emph{following the same sequence of sides to play on for these $r$ rounds}. Then, the principle of parallel play asserts that he can win the $r$-round {\ms} game on $(\cA, \cB)$ as well, just by following his strategies on the two sub-games simultaneously. The fact that he follows the same sequence of sides to play on for the $r$ rounds means that there is never any ambiguity about which side he needs to play on in the (larger) game.

Parallel play turned out to be a powerful technique to save redundant moves for Spoiler. Using Theorem \ref{thm:MSfundamental}, this led to showing  in \cite{parallelplay} that many natural FO-properties on \emph{ordered} structures can be expressed using exponentially fewer quantifiers than the na\"{i}ve upper bound.
%\jon{By who -- us? Perhaps not worth saying this.} \rik{Agree, changing. Is the current version of that sentence better?}

Since we will use parallel play in this paper, we formally state it as a lemma below.

\begin{restatable}[Parallel Play Lemma, \cite{parallelplay}]{lemma}{parallelplay}\label{lem:parallelplay}
Let $\cA$ and $\cB$ be two sets of pebbled structures, and fix $r \in \mathbb{N}$. Suppose that $\cA$ and $\cB$ can be partitioned as $\cA = \cA_1 \sqcup \cA_2$ and $\cB = \cB_1 \sqcup \cB_2$ in such a way  that $\bS$ has a winning strategy for the $r$-round MS game on $(\cA_i,\cB_i)$, where $i=1,2$, satisfying the following conditions:
\begin{enumerate}
\item Both strategies follow the same sequence $(S_1, \ldots, S_r)$ of left or right sides that Spoiler plays on, where $S_i \in \{L, R\}$ for $1 \leq i \leq r$.
\item At the end of the sub-games, both of the following are true:
\begin{itemize}
    \item There do not exist pebbled structures from $\cA_1$ and $\cB_2$ forming a matching pair.
    \item There do not exist pebbled structures from $\cA_2$ and $\cB_1$ forming a matching pair.
\end{itemize}
\end{enumerate}
Then $\bS$ wins the $r$-round {\ms} game on $(\cA,\cB)$ with the same sequence $(S_1, \ldots, S_r)$ of sides he plays on.
\end{restatable}

%% file: 03upper.tex
\section{The Upper Bound}
\label{sec:upper}

Suppose that the schema $\tau=(R_1,\ldots,R_t,c_1,\ldots,c_s)$ is such that the maximum arity of the relation symbols $R_1,\ldots,R_t$ is $r$,  where  $r\geq 2$.
 We start with a lemma, which establishes an upper bound on the number of FO-sentences over $\tau$ (up to logical equivalence) with $m$ quantifiers.

\begin{restatable}{lemma}{countsentences}\label{lem:countsentences}
Let  $\tau=(R_1,\ldots,R_t,c_1,\ldots,c_s)$ be a schema such that the maximum arity of the relation symbols in $\tau$ is $r$, where  $r\geq 2$. Then, up to logical equivalence, the number of \emph{FO}-sentences over $\tau$  with $m$ quantifiers is at most $2^{m + 2^{(t+1)(m + s)^r}}$, and each such sentence has a quantifier-free part with at most $(t + 1)(m + s)^r\cdot 2^{(t + 1)(m + s)^r}$ (positive or negative) atomic formulas.
%\jon{Since this lemma is subsequently appealed to both for its upper bound on the number of logically equivalent sentences and the total number of types I would include the upper bound on the total number of types as part of the statement of the lemma.}
%\rik{I'll ignore this comment, since we don't want to define types yet.}
%\ron{The lemma needs a statement of the length of the sentence, which is stated in the proof, and made use of later.}
\end{restatable}

\begin{proof}
Fix $m$, and let $x_1, \ldots, x_m$ be $m$ distinct
variables. For each $i$, the number of atomic formulas using the relation symbol $R_i$ (say of arity $r_i$) is $(m + s)^{r_i} \leq (m + s)^r$. Furthermore, there are $(m+s)^2$ atomic formulas involving the equality symbol $=$. Therefore, the total number of atomic formulas over $\tau$ is at most $(t+1)(m + s)^r$.

By definition, a \emph{type} over $\tau$ is a conjunction of the form $\theta_1  \land \ldots \land  \theta_d$, where $d \leq (t+1)(m + s)^r$ by the argument above, and each $\theta_i$ is either an atomic formula or a negated atomic formula, and for each atomic formula, either the formula itself or its negation appears exactly once as a conjunct. Therefore, the total number of types is at most $2^{(t+1)(m + s)^r}$ (corresponding to each choice of positive or negative appearance of each atomic formula).

Now, every FO-sentence $\varphi$ in prenex normal form with $m$ quantifiers is logically equivalent to an FO-sentence $\psi$ of the form $Q_1x_1\ldots Q_mx_m\chi$, where $\chi$ is a disjunction of types. Note that the quantifier-free part $\chi$ has at most $(t+1)(m + s)^r \cdot 2^{(t+1)(m + s)^r}$ atomic formulas (positive or negative), since each type has at most $(t+1)(m + s)^r$ atomic formulas.
%\jon{This is really an $O(\cdot)$ statement since we have a disjunction of at most $2^{(t+1)(m + s)^r}$ types, each of length at most $(t+1)(m + s)^r$, so the true length is  $(t+1)(m + s)^r \cdot 2^{(t+1)(m + s)^r}$.}
Also, there are $2^m$ choices for the quantifier block $Q_1x_1\ldots Q_mx_m$, and there are $2^{2^{(t+1)(m + s)^r}}$ choices for the subset of types needed to define $\chi$. By combining these two facts, we obtain the desired result.
\end{proof}

We now establish the $\LOGSPACE$ upper bound for $\winms_m$, for each fixed $m\geq 1$.

\begin{restatable}{proposition}{winmsmupper}\label{prop:winmsmupper}
For each fixed $m \geq 1$, the problem $\winms_m$ is in $\LOGSPACE$.
\end{restatable}

\begin{proof}
Assume first that we have unpebbled structures, and so we are looking for separating \emph{sentences} rather than formulas with free variables. When $m \geq 1$ is fixed, the number  of pairwise inequivalent sentences over $\tau$ with $m$ quantifiers is constant; this follows directly from Lemma \ref{lem:countsentences}, by observing that the entire expression for the upper bound is a constant since $\tau$ and $m$ are fixed. Furthermore, each such sentence is equivalent to one of the form $\psi = Q_1x_1\ldots Q_mx_m\chi$, where $\chi$ has constant length (also by Lemma \ref{lem:countsentences}, since $t$, $s$, $r$, and $m$ are all constants).
%\ron{This is a portion of the length given by the previous lemma's proof.}

Consider the following model-checking problem: given a finite $\tau$-structure $\bA$, and a FO-sentence $\psi$ of the form above, is it the case that $\bA\models \psi$? This problem can be solved by cycling through each of the $|A|^m$ possible instantiations of the quantified variables $Q_1x_1\ldots Q_mx_m$, and examining whether each such instantiation satisfies one of the types in the quantifier-free formula $\chi$.
%\jon{This seems to be stated a bit incorrectly. You are cycling through the $|A|^m$ different  instantiations of the variables $x_1,...,x_m$, not instantiations of the quantifier block -- which I would take to mean instantiations of the quantifiers in the block. These are of course set by the definition of $\psi$.}
This is easily seen to be a $\LOGSPACE$ protocol, since we only need: (i) space to store each instantiated tuple of $m$ variables, which takes space $O(m\log(|A|))$; (ii) space to keep a counter to cycle through all the instantiations of the $|A|^m$ tuples, which takes space $O(m\log(|A|))$; and (iii) space to check the that the instantiated tuple satisfies the constant-length quantifier-free formula $\chi$, which is well-known to be in $\LOGSPACE$.

Therefore, given an instance $(\cA, \cB)$ of $\winms_m$, we can consider each of the FO-sentences 
%\jon{s/b `sentence' or `each of the FO-sentences'}
$\psi$ in prenex normal form with $m$ quantifiers (there are a constant number of such sentences), and test whether, for each $\bA \in \cA$, we have $\bA \models \psi$, and for each $\bB \in \cB$, we have $\bB \models \lnot\psi$. As argued above, this takes space logarithmic in the maximum size of a structure in $\cA\cup\cB$.

When the input instance has pebbled structures, the proof goes through in exactly the same way, by considering pairwise inequivalent formulas over $\tau$ with free variables corresponding to the pebbles played in the input instance.
\end{proof}

The same ideas can be adapted to yield an upper bound for $\winms$.
    
\begin{restatable}{proposition}{winmsupper}\label{prop:winmsupper}
The problem $\winms$ is in $\NEXP$.
\end{restatable}

\begin{proof}
We can follow similar ideas as in the proof of Proposition \ref{prop:winmsmupper}. Once again, assuming unpebbled structures first, we need to only consider the distinct FO-sentences $\psi$ (up to logical equivalence) with $m$ quantifiers, but this time the number of these sentences is not constant, since $m$ is part of the input. Nevertheless, as before, we know from Lemma \ref{lem:countsentences} that there are at most $2^{m + 2^{(t+1)(m + s)^r}}$ such sentences.

Now, consider our model-checking problem once again: given a finite $\tau$-structure $\bA$, and a FO-sentence $\psi$ of the form above, is it the case that $\bA \models \psi$? Again, by cycling through the instantiations of the quantified variables, and examining whether each such instantiation satisfies one of the (at most) $2^{(t+1)(m + s)^r}$ types in the quantifier-free formula $\chi$, we can conclude that the model-checking problem takes deterministic time $O(|A|^m2^{(t+1)(m + s)^r})$ for any $\tau$-structure $\bA$.

So now, consider the following nondeterministic algorithm for verifying that $\bS$ wins an instance $(\cA, \cB)$ of the $m$-round {\ms} game.

\begin{enumerate}
    \item Guess a FO-sentence $\psi$ with $m$ quantifiers. Without loss of generality, assume that $\psi$ is of the form $Q_1x_1\ldots Q_mx_m\chi$, where $\chi$ is a disjunction of conjunctions $\theta_1\land \cdots\land \theta_d$, with $d \leq (t+1)(m + s)^r$, as in the proof of Lemma \ref{lem:countsentences}.
    \item For each $\bA \in \cA$ and each $\bB \in \cB$, check that $\bA \models \psi$ and $\bB\not\models \psi$.
 \end{enumerate}
The guess of $\psi$ is a non-deterministic step that takes time $O(2^{(t+1)(m + s)^r})$. Each subsequent model-checking step takes time $O(N^m2^{(t+1)(m + s)^r})=O(2^{m\log(N)+(t+1)(m + s)^r})$, where $N$ is the maximum size of a structure in $\cA \cup \cB$. Thus, the second step of the algorithm takes time $O(|\cA|\cdot|\cB|\cdot 2^{m\log(N)+(t+1)(m + s)^r})$. Therefore, this is a nondeterministic exponential time algorithm. 

Once again, the proof goes through similarly in the case of pebbled structures (by guessing a separating formula this time with free variables corresponding to the pebbles placed in the input instance).
\end{proof}

%\jon{I'm not sure how important this is, but note that both of these proofs have been given with the assumption of a separating \emph{sentence} and, hence, for the case of initially unpebbled sets of structures.} \rik{Good catch! It's not important for the purposes of the proof at all, but the writing needs to change in very minor places; I shall do this in a later pass.}

We remark here that the two ``obvious'' approaches for showing an upper bound for $\winms$ -- either simulating the game, or guessing a separating sentence -- both seem to result in an exponential blow-up of the problem, resulting in a $\NEXP$ upper bound at best. If we try to simulate the game directly, even on the game instances starting with singleton sets $(\{\bA\}, \{\bB\})$, the number of copies on either side goes up in general to $|A|^{O(m)}$ by the end of the game.
%\jon{I believe that number should be $|A|^m$ (resp., $|B|^m$) -- see our paper https://lmcs.episciences.org/15000/pdf, Theorem 3.9. One could guess a Spoiler strategy and then would have to test all pairs of boards for an isomorphism. You end up with a slightly different expression; using the prior notation I think it is $O(N^m)$, but still NEXPTIME.}
%\rik{Yes yes, I agree, that was a typo from the early draft where I was using $r$ for the number of rounds.}
Therefore, the number of objects to keep track of seems to grow exponentially in the length of our input. Guessing a separating sentence achieves $\NEXP$, as shown in Proposition \ref{prop:winmsupper}. In each case, this seems to happen primarily because of Duplicator's ability to play obliviously; when she does so, she increases the number of possibilities exponentially, and this seems to be an inherent attribute of the problem.

We remark here that Hella and Luosto \cite{HL24} have introduced several different games that capture the number of quantifiers, hence each of them is  equivalent to the {\ms} game. None of these games, however, seems to yield a better upper bound than the one derived here, at least in a straightforward way. In particular,
one of these games is the {\sc Monotone Prefix Game} in which, unlike the {\ms} game, the Duplicator does \emph{not} have an oblivious strategy. A na\"{i}ve analysis of this game, however, only yields $\Sigma_2^{\rm exp}$ as an upper bound, where  $\Sigma_2^{\rm exp}$ is the second level of the exponential hierarchy (and contains $\NEXP$).
%a worse upper bound, 
%for \winms,
%namely, 
%{\winms} is in 
%$\Sigma_2^{\rm exp}$, the second-level of the exponential hierarchy, which contains $\NEXP$. %Therefore, this approach seemingly gets subsumed by Proposition \ref{prop:winmsupper}.

%% file: 04lower.tex
\section{The Lower Bound}
\label{sec:lower}

In this section, we prove our main results, through our lower bound constructions. For ease of exposition, we shall do this in two parts; first, we will prove that the $\winms$ problem is $\NP$-hard, which shall contain all the essential ingredients of the $\PSPACE$-hardness proof, albeit on slightly simpler gadgets. In particular, this proof will illustrate the main ideas from inapproximability in complexity theory that we will use. We shall then prove $\PSPACE$-hardness through substantially more involved gadgets.

\subsection{Warm-Up: NP-Hardness of EF Games}\label{sec:EFNP}

We start by providing a self-contained proof that $\winef$ is $\NP$-hard, a result already shown in \cite{PezzoliThesis}. We first recall the rules of the {\ef} (EF) game briefly.

An instance of the EF game consists of a natural number $m$, and a pair of finite pebbled $\tau$-structures $\langle\bA ~|~ a_1, \ldots, a_k\rangle$ and $\langle\bB ~|~ b_1, \ldots, b_k\rangle$, each pebbled with the same subset $\{x_1, \ldots, x_k\}$ of colors. We shall now describe a run of the $m$-round EF game on $(\langle\bA ~|~ a_1, \ldots, a_k\rangle, \langle\bB ~|~ b_1, \ldots, b_k\rangle)$, by describing inductively what happens in round $i$, for $1 \leq i \leq m$. Assume that in rounds $1$ through $i - 1$, pebbles of colors $x'_1, \ldots, x'_{i-1}$ have been played on the elements $a'_1, \ldots, a'_{i-1} \in A$ respectively, and on elements $b'_1, \ldots, b'_{i-1} \in B$ respectively, resulting in the two pebbled structures $\langle \bA ~|~ a_1, \ldots, a_k, a'_1, \ldots, a'_{i - 1}\rangle$ and $\langle \bB ~|~ b_1, \ldots, b_k, b'_1, \ldots, b'_{i - 1}\rangle$ at the end of round $i - 1$ (and hence, at the beginning of round $i$). In round $i$, $\bS$ chooses one of these pebbled structures,
%\jon{I would change `them' => 'the pebbled structures'; the word `them' does not have a completely unambiguous referent}
say, the first one, i.e., $\langle \bA ~|~ a_1, \ldots, a_k, a'_1, \ldots, a'_{i - 1}\rangle$.
%\jon{I would change `the first one' to $A$ since, again, `first one' is not 100\% unambiguous.}
He then plays a pebble of an \emph{unused} color $x'_i$ on an element $a'_i \in A$, creating $\langle \bA ~|~ a_1, \ldots, a_k, a'_1, \ldots, a'_i\rangle$.

%\jon{Even though we don't concern ourselves with such details, in EF games one can reuse pebbles of the same color -- perhaps mention this in a footnote.}
In response, $\bD$ plays a pebble of color $x'_i$ on an element $b'_i \in B$, creating the pebbled structure $\langle \bB ~|~ b_1, \ldots, b_k, b'_1, \ldots, b'_i\rangle$. The run ends after $m$ rounds, at which points, $\bD$ wins this run if $\langle \bA ~|~ a_1, \ldots, a_k, a'_1, \ldots, a'_m\rangle$ and $\langle \bB ~|~ b_1, \ldots, b_k, b'_1, \ldots, b'_m\rangle$ form a matching pair, and $\bS$ wins otherwise. We define a winning strategy for $\bS$ and $\bD$ as in {\ms} games, and the notion of a separating formula also carries over analogously. As in the {\ms} game, we can adopt the convention of calling the two sides consistently \emph{left} and \emph{right} throughout. Finally, we have the fundamental theorem of EF games as well, showing that these games capture quantifier rank.

\begin{theorem}[Fundamental Theorem of EF Games, \cite{ehrenfeucht1961application, fraisse1950nouvelle}]\label{thm:EFfundamental}
    $\bS$ has a winning strategy in the $m$-round EF game on $(\langle\bA ~|~ a_1, \ldots, a_k\rangle, \langle\bB ~|~ b_1, \ldots, b_k\rangle)$ iff there is a separating formula (with $k$ free variables) for $(\langle\bA ~|~ a_1, \ldots, a_k\rangle, \langle\bB ~|~ b_1, \ldots, b_k\rangle)$ with quantifier rank at most $m$.
\end{theorem}

To show the $\NP$-hardness of $\winef$, we adopt a similar gadget to the one in \cite{PezzoliThesis}, with several important modifications, including a reduction from $\domset$ instead of one from $\textsc{Clique}$, which establishes the hardness even for schemas with predicates of arity at most $2$. However, first we need a few definitions.

\begin{definition}
Let $\langle\bA ~|~ \vec{a}\rangle, \langle\bB ~|~ \vec{b}\rangle$ be $\tau$-structures, and fix elements $u \in A$ and $v \in B$, and natural numbers $m, m'$. In the $m$-round EF game on $(\langle \bA ~|~ \vec{a}\rangle, \langle \bB ~|~ \vec{b}\rangle)$, we say $\bS$ \emph{forces} the pair $(u, v)$ (within $m'$ additional rounds) if $\bS$ has a strategy such that the following hold:
%\jon{Two minor comments: (i) I think you should replace `naturals' with `natural numbers' -- `naturals' seems to colloquial, and (ii) the way it is written, one could imagine that the response needs only be played in any of the remaining $m'$ moves and not as the immediate response, which is not what is intended.}
\begin{itemize}
\item he can play on $u \in A$, and if $\bD$ does not respond to this play by playing on $v \in B$, then she loses within $m'$ additional rounds;
\item he can play on $v \in B$, and if $\bD$ does not respond to this play by playing on $u \in A$, then she loses within $m'$ additional rounds.
\end{itemize}
Usually $m'$ is understood from the context, and not specified. Furthermore, $\bS$ can often force a pair without actually playing on that pair. We can extend this definition for the {\ms} game as well.
\end{definition}

Finally, we define the problem that we shall reduce from.

\begin{definition}
For an undirected graph $G = (V, E)$, a \emph{dominating set} is a subset $U \subseteq V$ such that every $v \in V - U$ has a neighbor in
$U$. Define the decision problem:
\begin{equation*}
    \domset := \{\langle G, k\rangle ~:~ G \mbox{ has a dominating set of size }k\}.
\end{equation*}
\end{definition}

It is well known that $\domset$ is a canonical $\NP$-complete problem \cite{DBLP:books/fm/GareyJ79}. Note also that $\domset$ is the decision problem underlying the optimization problem 
$\mindomset$: given an undirected graph $G=(V,E)$, find the minimum size of a dominating set of $G$ (and a dominating set $U$ of that size). The $\NP$-hardness of $\domset$ implies that there is no polynomial-time algorithm for computing the minimum value $\opt(G)$ of $\mindomset$ on $G$ unless $\P = \NP$.
In fact, it is known that $\mindomset$ is hard even to approximate. More precisely, suppose that $\mathscr{A}$ is an algorithm such that, given a graph $G$, the algorithm returns a value $\mathscr{A}(G)$ as an approximate value of $\opt(G)$. It is well known and not hard to see that if $\mathscr{A}$ is the greedy algorithm, then
$\mathscr{A}(G)/\opt(G)\leq 1+ \log(|V|)$.
%\ron{I don't offhand see this.}
%Furthermore, $\domset$ admits\footnote{ Strictly speaking, $\domset$ is a decision problem; the corresponding optimization problem takes as input an undirected graph $G$ and asks: ``what is the minimum $k$ for which $\langle G, k\rangle \in \domset$?'' An approximation algorithm for this decision problem takes $G$ as input, and outputs some $k'$, satisfying an approximation guarantee between the output $k'$ and the true value $k$.} the following inapproximability result \cite{domset1, domset2}, whose proof is omitted.
%end-commentout
However, the following inapproximability result has been established by Raz and Safra \cite{domset1}

\begin{theorem}[Inapproximability of Dominating Set, \cite{domset1}]\label{domsetapprox}
%It is impossible to approximate $\domset$ to within an $o(\log n)$ multiplicative factor unless $\P = \NP$.
Unless $\P = \NP$, there is a constant $c>0$ with the following property: there is no polynomial-time algorithm $\mathscr{A}$ that, given an arbitrary undirected graph $G=(V,E)$, returns an approximate value $\mathscr{A}(G)$ for $\mindomset$ on $G$ such that $\mathscr{A}(G)/\opt(G)\leq c\log(|V|)$.
\end{theorem}
As a consequence of Theorem \ref{domsetapprox},  no polynomial-time algorithm $\mathscr{A}$ and no constant $\varepsilon$ exist such that $\mathscr{A}(G)/\opt(G)\leq \varepsilon$, unless $\P=\NP$. 

In Section \ref{sec:MSNP}, we will leverage the  inapproximability results for $\mindomset$ to prove that $\winms$ is $\NP$-hard.  In what follows in this section, we will give a polynomial-time reduction of $\domset$ to $\winef$.

%The main idea of our proof will be to reduce $\domset$ to $\winms$. 
%\jon{Won't it be a reduction from a fixed $o(\log n)$ factor approximation of $\domset$ to $\winms$?} \rik{No, I don't quite know what you mean. A reduction is from a problem to another problem. We can use the reduction (as stated) to obtain an algorithm that obtains an approximation to the problem reduced from.} \jon{It is indeed a decision problem to find an $f(n)$-APX to a problem such as the smallest DOMINATING-SET. If you had a reduction from DOMINATING SET to WIN(MS) you would be done, without needing to worry about an approximation.} \rik{Hmm okay, I didn't know this terminology/convention, I somehow thought this was some sort of many-to-one reduction. Phokion/Neil, what's the most precise way to state this approach?}
%\phokion{There is some confusion here between the optimization problem and the decision problem. We should state the optimization problem, since the inapproximability result is about the optimization problem, not the decision problem. We should also mention the underlying decision problem. I can take care of this later today or tomorrow.}
%end-commentout
Given an instance $\langle G, k\rangle$ of $\domset$, where $G$ has $n$ vertices $\{v_1, \ldots, v_n\}$, we will build a $\taug$-structure $\bA$ with vertices $a, a' \in A$ such that:
\begin{equation}
    \langle G,k \rangle \in \domset \iff \langle k+1, (\langle \bA ~|~ a\rangle, \langle \bA ~|~ a'\rangle)\rangle \in \winef.
\end{equation}

Here, schema $\taug$ has a single binary relation symbol $E$, and thus a $\taug$-structure is just a directed graph. For the sake of clarity in exposition, for now we imagine $\taug$ is augmented with three unary relation symbols $R$, $G$, and $B$ (thought of as red, green, and blue vertices respectively). We shall see later on how to remove these from the schema. We now describe the building blocks for $\bA$.

\paragraph*{The Building Block Gadgets}

 As shown in Fig.~\ref{figIj}, gadget $\bI_j$ consists of a pair of red vertices $p$ and $p'$, a pair of blue vertices $q$ and $q'$, and a pair of green vertices $r$ and $r'$, with eight vertices ($c_i, d_i$, with $i \in \{1, 2, 3, 4\}$) between them, with the edge connections shown. The vertices $c_i$ and $d_i$ are distinguished by means of having a different number of additional uncolored in-neighbors; each $c_i$ has $j$ uncolored in-neighbors $\{a^i_1, \ldots, a^i_j\}$, and each $d_i$ has $j - 1$ uncolored in-neighbors $\{b^i_2, \ldots, b^i_j\}$. Note that this gadget $\bI_j$ is a generalized version of the gadgets used in \cite{CFI}. In particular, this gadget has a number of interesting automorphisms. Indeed, we can obtain automorphisms that keep any one of the three pairs $(p, p')$, $(q, q')$, and $(r, r')$ fixed, and permute the elements of both the other pairs (and re-labels the middle vertices $c_i$ and $d_i$ in the process). For instance, there is an automorphism of this gadget that fixes $p$ and $p'$, where $q \mapsto q'$, $q' \mapsto q$, $r \mapsto r'$, and $r' \mapsto r$.

\input{Images/Ij}

We refer to the $c_i$'s and $d_i$'s as \emph{vertices in the middle of the gadget $\bI_j$}.

The key property of $\bI_j$ that we will use is the following:

\begin{lemma}[Polarity-Preserving Lemma]\label{polarity}
For $j \geq 1$, $\bS$ can force the pair $(q,q')$ or $(r,r')$ in the $(j+1)$-round EF game on $(\langle~\bI_j ~|~ p~\rangle, \langle~\bI_j ~|~ p'~\rangle)$ (in fact, in round $2$), but not in the $j$-round game.
%\jon{I think you should preface by saying `For $j \geq 1$,' and also state that this forced move occurs in the 1st round.}\rik{ I fixed the $j$ comment. The forced move does not need to occur in the 1st round -- it's on the 2nd round unless Duplicator chooses to lose the game in $j + 1$ rounds through counting; how would you prefer that I write that? Maybe just ``in round $2$'' would suffice?} \jon{Right -- 2nd round; ``in round 2'' would suffice.}
\end{lemma}

\begin{proof}
In the $(j+1)$-round game, $\bS$ can play on $c_1$ in $\langle~ \bI_j ~|~ p ~\rangle$. $\bD$ must respond with an uncolored out-neighbor of $p'$, hence on $c_3$, $c_4$, $d_3$, or $d_4$. If she responds on $d_3$ or $d_4$ (say, $d_3$), $\bS$ can play his next $j$ moves on $a^1_1$ through $a^1_j$; $\bD$ can respond for the first $j - 1$ of these moves by playing on $b^3_2$ through $b^3_k$, but cannot respond in the last round while preserving an isomorphism. Therefore, for her first response, $\bD$ must play on $c_3$ or $c_4$. If she plays on $c_3$, $\bS$ plays on $r$. $\bD$ must play on a green out-neighbor of $c_3$ in $\langle~\bI_j ~|~ p'~\rangle$, which is $r'$, so the pair $(r, r')$ is forced. If she plays instead on $c_4$, $\bS$ plays on $q$. $\bD$ must play on a blue out-neighbor of $c_4$ in $\langle~\bI_j ~|~ p'~\rangle$, which is $q'$, and so $(q, q')$ is forced.

In the $j$-round game, if $\bS$ plays on $c_1$, $\bD$ can respond on $d_3$ and match $\bS$'s second round move. $\bS$ cannot exhibit the difference between the $c$'s and $d$'s in just $j - 1$ remaining rounds, so he loses. Similarly, if $\bS$ plays on $c_2$ for his first move, $\bD$ can respond on $d_4$. If $\bS$ plays
on $d_1$, $\bD$ can respond on $c_3$. And if $\bS$ plays on $d_2$, $\bD$ can respond on
$c_4$. In all these cases, $j$ more rounds are not enough to exhibit the difference between the $c$'s
and $d$'s. If $\bS$ tries anything else, $\bD$ responds symmetrically, using the automorphisms described above before Lemma \ref{polarity}, and wins each time.
\end{proof}

\paragraph*{The Reduction Instance}

Now, given the $\domset$ instance $\langle G, k\rangle$, we construct the auxiliary structure $\tilde{\bA}$ as depicted below (Fig.\ \ref{figX}). Note that the $\domset$ instance is on an \emph{undirected} graph, but the reduction will use directed edges (since the relation symbol $E$ in $\taug$ is interpreted as directed edges).
%\jon{There is again a problem with Figure numbering -- this should be Figure 2. Also, since the gadgets are directed, this might be a good time to mention that the DOMINATING-SET instance is with respect to an \emph{undirected} graph, $G$. DOMINATING-SET can also be considered on a directed graph.}
This structure is formed by stacking $k$ gadgets, consisting of a copy of $\bI_k$, then a copy of $\bI_{k-1}$, and so on, all the way down to a copy of $\bI_1$, and then adding edges $(q_j, p_{j-1}), (q'_j, p'_{j-1}), (r_j, p_{j-1}), (r'_j, p'_{j-1})$ for all $2 \leq j \leq k$. (Here, we refer to the $p$'s, $q$'s, and $r$'s in the block $\bI_j$ with the subscript $j$ to resolve ambiguities.) We name the vertices in the middle of the block $\bI_j$, $c^j_i$ or $d^j_i$, depending on whether they are adjacent to $j$ in-neighbors or $j - 1$ in-neighbors, respectively. Note that $\tilde{\bA}$ has not used the input graph $G$ yet; we use it in the next step of the construction.
%\jon{This auxiliary structure, $\tilde{\bA}$, seems to have nothing to do with the DOMINATING-SET instance, contrary to the lead-in sentence of the paragraph.}\rik{It uses $k$. The next structure $\bA$ will actually use the edges of the input graph $G$ as well.} \jon{OK -- perhaps say that $\tilde{\bA}$ uses just $k$ and that everything else about the graph will be incorporated in the next stage of the construction.}

\input{Images/auxX}

To build $\bA$ from $\tilde{\bA}$, we do the following modifications. In what follows, whenever we \emph{replace} a vertex $v$ by an independent set $S_v$, this means that:
\begin{itemize}
    \item every edge $(u,v)$ incoming into $v$ is replaced by the \emph{set} $\{(u,w) ~:~ w \in S_v\}$ of edges, and
    \item every edge $(v,u')$ outgoing from $v$ is replaced by the \emph{set} $\{(w, u') ~:~ w \in S_v\}$ of edges.
\end{itemize}
Note that the set $S_v$ itself gets no internal edge.
%\jon{Presumably this is also true for the out-edges.}
Recall that $G$ has $n$ vertices, $\{v_1, \ldots, v_n\}$.
\begin{enumerate}
    \item The vertices $v$ in the middle of $\bI_j$ are replaced by size-$n$ independent sets $S_v := \set{\langle{v,v_t}\rangle ~:~ v_t \in V(G)}$.
    For instance, if $v_6$ is a vertex of $G$, then we have a vertex such as $\angle{c^{17}_1, v_6}$, which corresponds to vertex $v_6$ in the size-$n$ independent set created by replacing the $c_1$ vertex in the $17$th gadget from the bottom.
    \item The blue (resp.~green) vertex $q_1$ (resp.~$r_1$) is replaced by the size-$n$ independent set $S_{q_1}$ (resp.~$S_{r_1}$) of blue (resp.~green) vertices, defined as follows:
    \begin{align*}
        S_{q_1} := \set{\langle q_1, v_t\rangle ~:~ v_t \in V(G)\rangle} && S_{r_1} := \set{\langle r_1, v_t\rangle ~:~ v_t \in V(G)\rangle}
    \end{align*}
    \item The blue (resp.~green) vertex $q'_1$ (resp.~$r'_1$) is replaced by the size-$(n + 1)$ independent set $S_{q'_1}$ (resp.~$S_{r'_1}$) of blue (resp.~green) vertices, defined as follows:
    \begin{align*}
        S_{q'_1} := \set{\langle q'_1, v_t\rangle ~:~ v_t \in V(G)\rangle} \cup \{ \mathsf{null}_{q'}\} && S_{r'_1} := \set{\langle r'_1, v_t\rangle ~:~ v_t \in V(G)\rangle} \cup \{\mathsf{null}_{r'}\}.
    \end{align*}
    \item We get rid of the red vertices $p_j, p'_j$ for $j = 1, \ldots, k - 1$ by connecting each of their in-neighbors with each of their corresponding out-neighbors. For example, for every in-neighbor $u$ of $p_j$ and for every out-neighbor $v$ of $p_j$, we add the edge $(u, v)$, and finally delete the vertex $p_j$ altogether.
    \item Finally, we add some more edges. For each vertex $v_t$ in $G$, we add all edges of the form $(\langle w, v_t\rangle, \langle w', v_t\rangle)$ for all $w \in \{q_1, q'_1, r_1, r'_1\}$, 
    and $w' \in \{c^j_1, \ldots, c^j_4, d^j_1, \ldots, d^j _4 ~:~ 1 \leq j \leq k\}$. Every vertex $v_t$ in $G$, therefore, contributes in this way to exactly $32k$ new edges.  For each edge $(v_s, v_t)$ in $G$, we add all edges of the form $(\langle w, v_s\rangle, \langle w', v_t\rangle)$ as well as $(\langle w, v_t\rangle, \langle w', v_s\rangle)$ for all $w \in \{q_1, q'_1, r_1, r'_1\}$, and $w' \in \{c^j_1, \ldots, c^j_4, d^j_1, \ldots, d^j_4\}$. Every edge $(v_s, v_t)$ in $G$, therefore, contributes in this way to exactly $64k$ new edges. Even with these additions, the special blue vertex $\mathsf{null}_{q'}$ and the special green vertex $\mathsf{null}_{r'}$ have no out-neighbors. 

\end{enumerate}

Furthermore, let $a = p_k$, and $a' = p'_k$. We can now state and prove our first main result.

\begin{restatable}{theorem}{EFNPhard}\label{thm:EFNPhard}
The problem $\winef$ is $\NP$-hard.    
\end{restatable}

\begin{proof}
As stated, we shall reduce from $\domset$. Given an instance $\langle G, k\rangle$ of $\domset$, we construct the instance $\langle k+1, (\langle \bA ~|~ a\rangle, \langle \bA ~|~ a'\rangle)\rangle$ as described above. Of course, it is easy to check that this instance can be constructed from $\langle G,k\rangle$ in polynomial time. In fact, this is a first-order mapping. Thus, we are presenting a first-order reduction from $\domset$ to $\winms$.

We shall now show that:
\begin{equation*}
    \langle G, k\rangle \in \domset \iff \langle k+1, (\langle \bA ~|~ a\rangle, \langle \bA ~|~ a'\rangle)\rangle \in \winef.
\end{equation*}
\begin{itemize}
    \item \textbf{Dominating Set Implies Spoiler Wins:} Suppose that $G$ has a dominating set of size $k$, say $U = \{u_1, \ldots, u_k\}$. We wish to show that $\bS$ has a winning strategy in the $(k+1)$-round $\ef$ game on $(\langle \bA ~|~ a\rangle, \langle \bA ~|~ a'\rangle)$. $\bS$ plays his first round move on $\langle c^k_1, u_1\rangle$. $\bD$ must respond on an uncolored out-neighbor of $a'$ on the right side, hence on one of the vertices in $S_w$ for $w \in \{c^k_3, c^k_4, d^k_3, d^k_4\}$. By the proof of Lemma \ref{polarity}, we know that any response by $\bD$ in $S_{d^k_3}$ or $S_{d^k_4}$ results in her losing the game in the remaining $k$ rounds. Therefore, $\bD$ must respond on a vertex in $S_{c^k_3}$ or $S_{c^k_4}$, forcing the pair $(r_k, r'_k)$ or $(q_k, q'_k)$ respectively.
    %\jon{We are never completely clear what we mean by the words \emph{forcing a pair $(a,b)$}; this last sentence appears to be the first time the notion is used.}
    %\jon{It seems that despite the fact that we are saying `forcing the pair $(r_k, r'_k)$ or $(q_k, q'_k)$' these forced moves never get played and indeed don't need to be played.} 
    $\bS$ can now play his second round move on $\langle c^{k-1}_1, u_2\rangle$. Again, $\bD$ must respond on a vertex in $S_w$ for some $w \in \{c^{k-1}_3, c^{k-1}_4, d^{k-1}_3, d^{k-1}_4\}$, since any other response would result her losing in the next round by $\bS$ playing on $q_k$ or $r_k$ (depending on whether $(q_k, q'_k)$ or $(r_k, r'_k)$ was forced). Again, if $\bD$ responds in $S_{d^{k-1}_3}$ or $S_{d^{k-1}_4}$, then $\bS$ wins in the remaining $k - 1$ rounds. It can be shown inductively that $\bS$ can keep going with this strategy, playing in round $i$ on $\langle c^{k+1-i}_1, u_i\rangle$, and $\bD$ must respond in one of the vertices in $S_w$ for $w \in \{c^{k+1-i}_3, c^{k+1-i}_4\}$, forcing the pair $(r_{k+1-i}, r'_{k+1-i})$ or the pair $(q_{k+1-i}, q'_{k+1-i})$ respectively. Hence, at the end of the $k$-th round, we know that one of the pairs $(q_1, q'_1)$ or $(r_1, r'_1)$ is forced in the block $\bI_1$, say WLOG $(q_1, q'_1)$. Finally, for his $(k+1)$-th round move, $\bS$ plays on the right side, on the vertex $\mathsf{null}_{q'}$. $\bD$ must respond on a vertex in $S_{q_1}$ on the left side, say on the vertex $\langle q_1, v_\ell\rangle$.
    
    As we noted before, the special vertex $\mathsf{null}_{q'}$ has no out-neighbors. We claim that $\langle q_1, v_\ell\rangle$ has an out-neighbor with a pebble on it, and hence the two structures cannot form a matching pair. The element $v_\ell$ corresponds to a vertex in $G$. Either it is in the dominating set $U$, or it has a neighbor in $U$. In the first case, say, $v_\ell = u_t$ for some $1 \leq t \leq k$. In round $t$, $\bS$ played on $\langle c^{k+1-t}_1, u_t\rangle$ on the left, and there is an edge from $\langle q_1, v_\ell\rangle$ to $\langle c^{k+1-t}_1, u_t\rangle$ by construction. In the second case, a similar argument holds, where $u_t$ is the neighbor of $v_\ell$ in the dominating set $U$.
    
    \item \textbf{No Dominating Set Implies Duplicator Wins:} Suppose that $G$ has no dominating set of size $k$. We wish to show that $\bD$ wins the $(k+1)$-round EF game on $(\langle \bA ~|~ a\rangle, \langle \bA ~|~ a'\rangle)$.
    
    We can show by induction on $i$ that if $\bS$'s first $i$ moves are not successively on vertices in the independent sets arising from the middle of blocks $\bI_k, \ldots, \bI_{k + 1 - i}$, then $\bD$ can maintain a matching pair in the remaining $k + 1 - i$ rounds.
    
    To see this for $i = 1$, note that any move by $\bS$ on a blue or green vertex can be met by $\bD$ on the same blue or green vertex on the other side (i.e., $q_r$ met with $q_r$, $q'_r$ met with $q'_r$, etc). Breaking isomorphisms now is not possible without going ``up'' the structure $\bA$, and there are not enough rounds remaining for this (since the $c$'s and $d$'s in block $\bI_k$ are indistinguishable with $k - 1$ rounds remaining). Similarly, any play by $\bS$ on an in-neighbor of a $c$ or $d$ vertex can be met in the same way. For $i > 1$, the induction goes through in the same way, using the fact that the $c$'s and $d$'s in block $\bI_{k + 1 - i}$ and higher are indistinguishable in the remaining rounds.
    
    Assume now that $\bD$ always ``mirrors'' $\bS$'s moves; if in round $i$, $\bS$ plays in $S_{c^{k+1-i}_j}$ on the vertex $\langle c^{k+1-i}_j, v_t\rangle$, then $\bD$ responds on $\langle c^{k+1-i}_{5-j}, v_t\rangle$.
    
    Therefore, the only possibility is for $\bS$ to force the game into the block $\bI_1$ using the first $k$ rounds, as stated. He must, therefore, force the pair $(q_1, q'_1)$ or $(r_1, r'_1)$ with one round remaining. Say WLOG the forced pair is $(r_1, r'_1)$. What can $\bS$ do in the final round? It is clear that any move except in $S_w$ for $w \in \{q_1, q'_1, r_1, r'_1\}$ can be met trivially. Of course, a move in $S_{q_1}$ or $S_{q'_1}$ can be met with the same move on the other side, whereas a move on $\langle r_1, v_t\rangle$ or $\langle r'_1, v_t\rangle$ can be met trivially with $\langle r'_1, v_t\rangle$ and $\langle r_1, v_t\rangle$ respectively. It therefore only suffices to consider the possibility when $\bS$ makes his last play on $\mathsf{null}_{r'}$. $\bD$ must now respond on a green vertex in $S_{r_1}$, i.e., on a vertex of the form $\langle r_1, v_\ell\rangle$ for some $v_\ell \in V(G)$.
    
    Since by assumption, $G$ does not have a dominating set of size $k$, consider the vertices of $G$ designated by $\bS$ in the first $k$ rounds. His moves in these rounds have specified some subset of vertices in $G$ of size at most $k$. This cannot be a dominating set, and so there must be some vertex $v_\ell$ in $G$ outside this subset, which does not have a neighbor in the subset. $\bD$ can play her $(k+1)$-th round move on $\langle r_1, v_\ell\rangle$. By construction (and by $\bD$'s ``mirroring'' strategy), this does not have an out-neighbor among the vertices with pebbles placed on it, and hence, it maintains the matching pair.
\end{itemize}
\end{proof}

We remark here that we can easily remove the colors of the vertices, thereby establishing the hardness result for even directed graphs. In particular, we can get rid of the red vertices easily (since we start with a pebble on $a$ and $a'$ anyway). In order to distinguish between the blue and green vertices, we simply use a self-loop on each of the green vertices, and then consider the vertices as uncolored. This converts the entire reduction structure into a directed graph. In particular, this shows that the ternary relation in the schema used in the proof of the main result in \cite{PezzoliThesis} is not necessary.

\begin{corollary}
The problem $\winef$ is $\NP$-hard even on directed graphs.
\end{corollary}

\subsection{NP-Hardness of MS Games}\label{sec:MSNP}

In this section, we adapt the ideas from Section \ref{sec:EFNP} to show that $\winms$ is $\NP$-hard. This will illustrate the main ingredients of our $\PSPACE$-hardness proof, including the use of approximability and parallel play.

\begin{restatable}{theorem}{MSNPhard}\label{MS-NPhard}
The problem $\winms$ is $\NP$-hard. 
\end{restatable}

\begin{proof}
The idea once again is to use the same construction, $\bA$,
%\jon{It is actually slightly different -- you have brought back the lower level red nodes in the construction that follows.}
%\rik{Hang on, no, I haven't brought back the lower level red nodes; the odd numbered moves are in the middle vertices, the even numbered ones are in green or blue.}
as in Section \ref{sec:EFNP} and use $\domset$. In particular, given an instance $\langle G, k\rangle$ of $\domset$, we again construct the same $\taug$-structure $\bA$ and show that:
\begin{align}
    \langle G, k\rangle \in \domset &\implies \langle{2k + 1}, (\{\langle \bA ~|~ a\rangle\}, \{\langle \bA ~|~ a'\rangle\})\rangle \in \winms. \label{eq1}\\
    \langle G, k\rangle \notin \domset &\implies \langle{k+1}, (\{\langle \bA ~|~ a\rangle\}, \{\langle \bA ~|~ a'\rangle\})\rangle \notin \winms. \label{eq2}
\end{align}
We will then use the known inapproximability of $\domset$ (Lemma \ref{domsetapprox}) to obtain our result.

\begin{itemize}
    \item \textbf{Dominating Set Implies Spoiler Wins the $(2k + 1)$-Round Game:} Suppose again that $G$ has a dominating set of size $k$, say $U = \{u_1, \ldots, u_k\}$. We wish to show that $\bS$ wins the $(2k+1)$-round $\ms$ game on $(\{\langle \bA ~|~ a\rangle\}, \{\langle \bA ~|~ a'\rangle\})$. We will use the technique of parallel play, described in Lemma \ref{lem:parallelplay}. The main idea is to partition a game into multiple parallel instances of sub-games, each corresponding to a different isomorphism class, and then playing them in parallel.
    
    We start with the usual instance of this game, where we focus first on the top-most builder gadget, $\bI_k$. Here, the $c$-vertices (with $k$ in-neighbors) are drawn as squares, and the $d$-vertices (with $k - 1$ in-neighbors) are drawn as circles.

    $\bS$ now makes his first round move (with pebble $x_1$) on the left, playing on $\langle c^k_1, u_1\rangle \in S_{c^k_1}$. $\bD$ has the ability to make copies for different responses. The only responses that preserve a partial isomorphism are when she plays on $S_w$, for $w \in \{c^k_3, c^k_4, d^k_3, d^k_4\}$, so we focus on those copies. Note that all other copies she creates can be discarded by Lemma \ref{lem:discard}. Partition the remaining pebbled structures on the right as:
    \begin{align*}
        \cB_1 &:= \{\langle \bA ~|~ a', v\rangle ~:~ v \in S_{c^k_3}\} \\
        \cB_2 &:= \{\langle \bA ~|~ a', v\rangle ~:~ v \in S_{c^k_4}\} \\
        \cB_3 &:= \{\langle \bA ~|~ a', v\rangle ~:~ v \in S_{d^k_3}\cup S_{d^k_4}\}.
    \end{align*}

    \input{Images/MSround1}

    $\bS$ now plays on the right (with pebble $x_2$), on these copies. On every pebbled structure in $\cB_1$, he plays on $r'_k$. On every pebbled structure in $\cB_2$, he plays on $q'_k$. And on every pebbled structure in $\cB_3$, he plays on one of the uncolored in-neighbors of the element with pebble $x_1$ on it. Redefine $\cB_1, \cB_2, \cB_3$ to these new sets (with the round-$2$ pebbles).
    
    Now, $\bD$ responds on the left in all possible ways, but the only responses to consider are when she plays on $q_k$, $r_k$, or an uncolored in-neighbor of the element with pebble $x_1$ (all other copies can be discarded by Lemma \ref{lem:discard}). Crucially, her responses on the left on $q'_k$ do not maintain matching pairs with any structures on the right (including those with pebble $x_2$ on $q'_k$), and her responses on the left on $r'_k$ also do not maintain matching pairs with any structures on the right (including those with pebble $x_2$ on $r'_k$). Partition the remaining pebbled structures on the left as:
    \begin{align*}
        \cA_1 &:= \{\langle \bA ~|~ a, \langle c^k_1, u_1\rangle, r_k\rangle\} \\
        \cA_2 &:= \{\langle \bA ~|~ a, \langle c^k_1, u_1\rangle, q_k\rangle\} \\
        \cA_3 &:= \{\langle \bA ~|~ a, \langle c^k_1, u_1\rangle, v\rangle ~:~ v \text{ is an in-neighbor of the element with pebble }x_1\}.
    \end{align*}

    See Fig.\ \ref{figRound2} for the remaining boards.

    \input{Images/MSround2}
    
    Consider the sub-games $(\cA_i, \cB_i)$ for $i \in \{1, 2, 3\}$, and note that no pebbled structures from $\cA_i$ and $\cB_j$ form a matching pair, for $i \neq j$. Furthermore, $\bS$ can follow the same strategy on each of these sub-games $(\cA_1, \cB_1)$ and $(\cA_2, \cB_2)$ for the next $2k - 3$ rounds as well, playing in round $2r - 1$ on the left, on vertex $\langle c^{k + 2 - 2r}_1, u_r\rangle$, and using round $2r$ to partition the right into three similar equivalence classes, for $1 \leq r \leq k - 1$. In the sub-game $(\cA_3, \cB_3)$, $\bS$ can choose any side to play on on any round, and just plays on more and more in-neighbors of the vertex with pebble $x_1$.
    %\jon{Doesn't the $k+1$st round have to be played by $\bS$ from the left; otherwise he cannot run the table, so to speak?}
    %\rik{Correct, also fixing in the next pass; either $\bS$ can play and switch sides earlier, or I can just have an extra round; the approximation will still go through.}
    By round $k + 1$, he will have won the sub-game $(\cA_3, \cB_3)$ just by running $\bD$ out of moves. Indeed, by round $2k$, he will have won all the sub-games where he wins by this ``counting'' technique.
    
    Finally, after round $2k$, $\bS$ is down to a large number of sub-games ($2^k$ sub-games, to be exact), but in each of them, either $(r_1, r'_1)$ or $(q_1, q'_1)$ is forced.

    In round $2k + 1$, $\bS$ plays on the right side with the same strategy as in the proof of Theorem \ref{thm:EFNPhard}. In the sub-games with $(q_1, q'_1)$ forced, he plays on $\mathsf{null}_{q'}$, and in the sub-games with $(r_1, r'_1)$ forced, he plays on $\mathsf{null}_{r'}$. Again, none of these vertices have any out-neighbors. $\bD$'s responses on the left (forced to be in $S_{q_1}$ or $S_{r_2}$) fail for the same reason as in the proof of Theorem \ref{thm:EFNPhard}: every response she can give has an out-neighbor on a previously played pebble (by $\bS$) on a previous round. Therefore, all matching pairs are gone.
    
    \item \textbf{No Dominating Set Implies Duplicator Wins the $(k+1)$-Round Game:} Suppose that $G$ has no dominating set of size $k$. We wish to show that $\bD$ wins the $(k+1)$-round $\ms$ game on $(\{\langle \bA ~|~ a\rangle\}, \{\langle \bA ~|~ a'\rangle\})$. But this is trivial from Theorem \ref{thm:EFNPhard}, since if $\bD$ wins an $\ef$ game on an instance, she certainly wins an $\ms$ game on the same instance. 
\end{itemize}

\paragraph*{Hardness Using Inapproximability}

Now, consider the following algorithm (Algorithm\ \ref{protocol}) that computes a $2$-approximation to $\domset$ given a graph $G$, using calls to $\winms$.

\begin{lstlisting}[caption={Computing a $2$-approximation to $\mindomset$.},label=protocol,captionpos=t,float, escapeinside={(*}{*)}]
On input (*$G$*):
    For (*$k = 1, \ldots, n$*):
        Fix the instance (*$\langle G, k\rangle$*) of (*$\domset$*);
        Run (*$\winms$*) on (*$\langle k, (\{\langle \bA ~|~ a\rangle\}, \{\langle \bA ~|~ a'\rangle\})\rangle$*);
        if (*$\winms$*) outputs YES:
            output (*$k-1$*);
\end{lstlisting}

This algorithm outputs the minimum $k$ such that $\bS$ wins the $(k+1)$-round game on $(\{\langle \bA ~|~ a\rangle\}, \{\langle \bA ~|~ a'\rangle\})$. To see that this works as claimed, note that, if $k'$ is the minimum size of a dominating set in $G$, then by equations \eqref{eq1} and \eqref{eq2}, $\bD$ wins the game instance with $k'$ rounds or fewer, while $\bS$ wins the game instances with $2k'+1$ rounds or more. Algorithm \ref{protocol} therefore outputs a number between $k'$ and $2k'$, which must, therefore, be a $2$-approximation to the optimum.
%since $\bS$ wins the $(k + 1)$-round game, there must be a dominating set in $G$ of size at least $k$ (by equation \eqref{eq2}). But since $\bD$ wins the $k$-round game, there is no dominating set in $G$ of size less than $(k/2)$ (by the contrapositive of equation \eqref{eq1}). Therefore, the minimum dominating set in $G$ has size between $k/2$ and $k$, and so we have a $2$-approximation. \jon{I found this explanation a bit confusing. To clear things up, I would add to the two parentheticals (by \emph{the contrapositive of} equation (\#)). An equivalent and perhaps more direct way to see that the Protocol gives a 2-APX is as follows: suppose the the minimum size DOMSET is $k$. Then $\bD$ will win games of size up to $k$, either player can win games of sizes $k+1,...,2k-1$ and $\bS$ will necessarily win games of size $2k$ and larger (by relations (2) and (3)). Thus, when $\bS$ wins for the first time, he will output a number between $k$ and $2k-1$ and his output is therefore a 2-APX.}

However, by Theorem \ref{domsetapprox} and the discussion following it, the problem $\domset$ is $\NP$-hard to approximate to any constant multiplicative factor, and so it follows that $\winms$ is $\NP$-hard as well.
\end{proof}

\subsection{The PSPACE Construction}

In this section, we present our main result, that of $\PSPACE$-hardness for the problem $\winms$. We shall use several ideas from Section \ref{sec:MSNP}, but the gadgets and the reduction problem will be different.

A \emph{quantified Boolean formula} (QBF) is an expression of the form $Q_1x_1\cdots Q_kx_k\psi$, where each $Q_i$ is either the existential quantifier $\exists$ or the universal quantifier $\forall$, and  $\psi$ is a propositional formula in CNF whose variables are among $x_1, \ldots, x_k$. Each such formula gives rise to a game, called the \emph{QBF game associated with $Q_1x_1\cdots Q_kx_k\psi$}, played between two players, $\mathsf{Player}_\exists$ and $\mathsf{Player}_\forall$, as follows: during a run of the game, in the $i$-th move, and depending on whether $Q_i$ is $\exists$ or $\forall$, $\mathsf{Player}_\exists$ or $\mathsf{Player}_\forall$ plays and assigns a truth value to the variable $x_i$. $\mathsf{Player}_\exists$ wins this run if the truth assignment produced by the two players at the end of the $k$-th move satisfies $\psi$; otherwise, $\mathsf{Player}_\forall$ wins the run. We say that the formula $Q_1x_1\cdots Q_kx_k\psi$ is \emph{true} if $\mathsf{Player}_\exists$  has a winning strategy that allows him to win every run of the QBF game in which she plays according to this strategy.
Stockmeyer \cite{DBLP:journals/tcs/Stockmeyer76} showed that the following problem, $\qbfsat$, is $\PSPACE$-complete: given a QBF formula $Q_1x_1\cdots Q_kx_k\psi$, is it true? In fact, this problem is $\PSPACE$-complete even when restricted to QBF formulas
$Q_1x_1\cdots Q_kx_k\psi$ in which $\psi$ is a 3CNF-formula.

Here, we are interested in the following variant of $\qbfsat$, which we call $\qsat$:
%\jon{This doesn't seem to me to be the most appropriate acronym. I would expect $\qsat$ to be just $\qbfsat$ where $\phi$ is constrained to be a 3CNF. Perhaps rename to $t-\qsat$.}
%\rik{Agree, will ask Phokion about it.}
given a QBF formula
$Q_1x_1\cdots Q_kx_k\psi$, where $\psi$ is a 3CNF-formula,  and a positive integer $t$, does $\mathsf{Player}_\exists$  have a strategy in the QBF game associated with $Q_1x_1\cdots Q_kx_k\psi$ such that at least $t$ clauses of $\psi$ are simultaneously satisfied at the end over every run? In other words:
\begin{align*}
    \qsat &:= \{\langle Q_1x_1\cdots Q_kx_k\psi, t\rangle ~:~ \psi \text{ has } t \text{ simultaneously satisfiable clauses}\}.
    \end{align*}

Clearly, $\qsat$ is a $\PSPACE$-complete problem. Furthermore, $\qsat$ is the decision problem underlying the optimization problem $\maxqsat$: given a quantified $3$-CNF formula $Q_1x_1\cdots Q_kx_k\psi$, find the maximum number of clauses that are simultaneously satisfiable under optimal play by both players.
The optimization problem $\maxqsat$ is $\PSPACE$-hard to approximate up to a fixed constant, as the following result (stated without proof) of Condon, Feigenbaum, Lund, and Shor \cite{DBLP:journals/cjtcs/CondonFLS95} asserts. The main ideas behind the proof can be found in Condon's excellent exposition \cite{condon1995approximate}.

\begin{theorem}[Inapproximability of Quantified SAT, \cite{DBLP:journals/cjtcs/CondonFLS95}]\label{qsatapprox}
Unless $\P = \PSPACE$, there is a constant $0 < c < 1$ with the following property: there is no polynomial-time algorithm $\mathscr{A}$ that, given an arbitrary quantified $3$-CNF formula $\varphi$, returns an approximate value $\mathscr{A}(\varphi)$ for $\maxqsat$ on $\varphi$ such that $\mathscr{A}(\varphi)/\opt(\varphi) \geq c$.
\end{theorem} 

We will show how to reduce $\maxqsat$ to $\winms$. Given a quantified $3$-CNF formula $Q_1x_1 \ldots Q_{2k}x_{2k}\psi$, we will build a directed graph $\bA$ with vertices $a, a' \in A$ such that an algorithm for $\winms$ used on $\langle t, (\{\langle \bA ~|~ a\rangle\}, \{\langle \bA ~|~ a'\rangle\})\rangle$ for a few (polynomially many) values of $t$ will give us a polynomial-time protocol for $(1 - \varepsilon)$-approximating $\maxqsat$ for arbitrarily small $\varepsilon$, which is impossible by Theorem \ref{qsatapprox} unless $\P = \PSPACE$.

\paragraph*{The Building Block Gadgets}

This time, we first define our building blocks $\bI_j$ using the even smaller gadgets $\bJ_j$ and $\bJ'_j$, shown in Fig.\ \ref{figJj}.

\input{Images/Jj}

These gadgets $\bJ_j$ and $\bJ'_j$ are once again Cai-F\"{u}rer-Immerman type gadgets that have lots of automorphisms, which $\bD$ can use to her advantage. We can think of these gadgets as ``connecting'' the node $z$ to the pair $(q, q')$, as shown in  Fig.\ \ref{figJj}. In our subsequent constructions, we shall often have nodes $u$, $v$, and $w$, and make a gadget $\bJ_j$ or $\bJ'_j$ connect $u$ to $(v, w)$; this will always mean taking the relevant gadget, and identifying the vertices $z \mapsto u$, $q \mapsto v$, and $q' \mapsto w$. We can now form our main building block $\bI_j$, shown in Fig.\ \ref{figIjPSPACE}.

\input{Images/IjPSPACE}

As depicted, $\bI_j$ consists of the following:
\begin{itemize}
    \item There are four vertices $\{p, p', q, q'\}$. Of these, $p$ and $p'$ have no in-neighbors, and $q$ and $q'$ have no out-neighbors.
    \item There are $16$ out-neighbors of $p$ in total. Of these, there are $8$ vertices $\{c_1, \ldots, c_8\}$ designated as stars in Fig.\ \ref{figIjPSPACE}, each of which is attached to $j$ independent in-neighbors (not shown), similar to the construction in Section \ref{sec:EFNP}. The other $8$ out-neighbors $\{d_1, \ldots, d_8\}$ of $p$ are designated as diamonds in Fig.\ \ref{figIjPSPACE}. Each of them is attached to $j - 1$ independent in-neighbors (not shown). These disjoint independent sets marking the difference between a $c$-vertex (i.e., a star) and a $d$-vertex (i.e., a diamond) are not shown in order to avoid clutter. The basic idea is that $\bS$ needs $j$ rounds of the {\ms} game to distinguish between a star and a diamond.
    \item There are $16$ out-neighbors of $p'$ in total as well. Of these, $4$ vertices $\{c_9, \ldots, c_{12}\}$ are stars, and the other $12$ vertices $\{d_9, \ldots, d_{20}\}$ are diamonds, as shown in Fig.\ \ref{figIjPSPACE}. Again, the (disjoint) independent sets of in-neighbors for each of these vertices are not shown.
    \item Each of the vertices in $\{c_1, c_2, c_3, c_4, d_1, d_2, d_3, d_4, c_9, c_{10}, c_{11}, c_{12}, d_{13}, d_{14}, d_{15}, d_{16}\}$ is connected to the pair $(q, q')$ by means of a \textcolor{blue}{$\bJ'_{j-1}$} gadget, shown in blue. The vertex $z$ of \textcolor{blue}{$\bJ'_{j-1}$} (from Fig.\ \ref{figJj}) is identified in each case with the $c_i$ or $d_i$ in question.
    %\jon{It seems that each of the vertices, $c_i, d_j$ take the place of $p$ in the respective $\bJ'_{j-1}$ gadgets. It took me a long time to realize this was allowed and that was what you were doing.} %\jon{It seems you mean the analog of $\bJ'_{j-1}$ where there are eight $c_i$'s and eight $d_i$'s, rather than four of each.}
    \item Each of the vertices in $\{c_5, c_6, c_7, c_8, d_5, d_6, d_7, d_8, d_9, d_{10}, d_{11}, d_{12}, d_{17}, d_{18}, d_{19}, d_{20}\}$ is connected to the pair $(q, q')$ by means of a \textcolor{red}{$\bJ_{j-1}$} gadget, shown in red. The vertex $z$ of \textcolor{red}{$\bJ_{j-1}$} (from Fig.\ \ref{figJj}) is identified in each case with the $c_i$ or $d_i$ in question.
    %\jon{These respective  $c_i$ and $d_j$ take the place of $p$ in the $\bJ_{j-1}$ gadgets, and, moreover, the vertices $q, q'$ are shared among all these gadgets.} %\jon{It seems you mean the analog of $\bJ_{j-1}$ where there are four $c_i$'s and twelve $d_i$'s, rather than four of each.}
\end{itemize}
This gadget is not ``too'' large; apart from the vertices $p, p', q, q'$, there are $12$ vertices of the form $c_i$ (with $j$ in-neighbors for each of them), $20$ vertices of the form $d_i$ (with $j - 1$ in-neighbors for each of them), and the middle vertices of the $32$ different \textcolor{red}{$\bJ_{j-1}$} and \textcolor{blue}{$\bJ'_{j-1}$} gadgets. Each of those smaller gadgets has $8$ middle vertices, with four of them having $j - 1$ in-neighbors, and the other four having $j - 2$ in-neighbors. This gives the total number of vertices in $\bI_j$ as:
\begin{equation*}
    4 + 32 + 12j + 20(j-1) + 32(8 + 4(j - 1) + 4(j-2)) = 288j - 112.
\end{equation*}
%\jon{After a lot of work and changing my assumptions, I managed to get this same number (!), bu I think it would help the reader if a precise accounting of each of these terms is given.}

Adapting the terminology we established earlier in this work, we refer to the $c_i$'s and $d_i$'s that are the out-neighbors of $p$ or $p'$ as \emph{upper middle vertices in $\bI_j$}, and we refer to the $c_i$'s and $d_i$'s \emph{within} the building blocks \textcolor{red}{$\bJ_j$}'s and \textcolor{blue}{$\bJ'_{j}$}'s as \emph{lower middle vertices in $\bI_j$}. We also refer to the other in-neighbors (i.e., those of the form $a_i^j$ or $b_i^j$) of any of the $c_i$ or $d_i$ vertices as \emph{distinguishing neighbors} for the corresponding $c_i$ or $d_i$.
%\jon{Presumably you mean the $a_i^j$ and $b_i^j$ in-neighbors; the $c_i, d_i$ vertices all have one additional in-neighbor.}

As in the construction in Section \ref{sec:EFNP}, this gadget $\bI_j$ also comes with its own notion of polarity preservation, in the following sense (see also Lemma 1 in \cite{Pezzoli98}).

\begin{lemma}[Polarity-Preserving Lemma]\label{polarity2}
For $j \geq 2$, $\bS$ can force the pair $(q,q')$ in the $(j+1)$-round {\ef} game on $(\langle~\bI_j ~|~ p~\rangle, \langle~\bI_j ~|~ p'~\rangle)$ (in fact, in round $3$), but not in the $j$-round game. (As in Lemma \ref{polarity}, the forced pebbles need not be explicitly played, just threatened.)
%\jon{The meaning of the notion of forcing a pair in a particular round seems a bit ambiguous since the pair (of pebbles) are never actually played. While I get the gist, it would be nice if what precisely is meant by this phrase could be spelled out. Also, it does not seem to me that the $j=1$ case works. In this case, the $J_0, J_0'$ gadgets have none of the $a_i^j, b_i^j$ nodes so I don't see how the pair $(q,q')$ is forced, even in $3$ moves, let alone in $2$ moves.}
\end{lemma}

\begin{proof}
In the $(j+1)$-round game, $\bS$ can play on the left side, on $c_5$ (or $c_6$, $c_7$, or $c_8$) of $\langle~ \bI_j ~|~ p ~\rangle$ in round $1$. $\bD$ must respond on an out-neighbor of $p'$ in $\langle~ \bI_j ~|~ p ~\rangle$. $\bD$ must respond with an uncolored out-neighbor of $p'$ in $\langle~ \bI_j ~|~ p' ~\rangle$. If she responds on any out-neighbor other than $c_9$, $c_{10}$, $c_{11}$, or $c_{12}$, she loses after $j$ more rounds, just by $\bS$ playing on the distinguishing neighbors of the elements just played (by showing the difference between $c_i$' and $d_i$). So assume without loss of generality that $\bD$ plays on $c_9$. Now, $\bS$ can play on the right side in round $2$, on an out-neighbor of $c_9$ (i.e., on a middle vertex of a \textcolor{blue}{$\bJ'_{j-1}$}). Note that $\bD$ is forced to respond on an out-neighbor of $c_5$ (i.e., on a middle vertex of a \textcolor{red}{$\bJ_{j-1}$}). $\bS$ now has essentially two\footnote{ This argument holds up to an ordering of the moves; the same forcing argument goes through for the only other options for $\bS$, which is to play first on the distinguishing neighbors for the $c_i$'s or $d_i$'s. We omit these details for the sake of brevity.} good options: he can play on $c_2$ or $c_4$ (and $\bD$ would have to respond on one of the $c$'s); or he can play on $d_1$ or $d_3$ (and $\bD$ would have to respond on one of the $d$'s).
%\jon{I believe $\bS$ could also play on any of the $a_1^2,...,a_{j-1}^2$ or $a_1^2,...,a_{j-1}^4$ and also force $(q, q')$.}
%\rik{I'm trying to visualize which $a$'s you mean, but I think either way, he would fall one move short. Is that not true?}
%\jon{I don't believe so, at least for $j \geq 2$ or $3$. Suppose $\bS$ plays on an $a_{j-1}^2$ on the right side. If $\bD$ plays on a $c_i$ or $d_i$ she immediately breaks a partial iso and loses. If she plays on $q$ or $q'$ then $\bS$ plays on $c_2$ on the right and again wins. Hence, $\bD$ must play on an $a_h^i$ or a $b_h^i$. If a $b_h^i$ is played, $\bS$ plays $c_2$ followed by the remaining $a_h^2$s, winning in exactly $j+1$ rounds. On the other hand, if $\bD$ plays on an $a_h^i$, $\bS$ will play on $c_2$ and $\bD$ will be forced to play on $c_i$, forcing the pair $(q, q')$. Note though, that $\bS$ does not need to play on $c_2$ to force play at $(q,q')$, so it is essentially the same ``forcing''.}
Both of these possibilities leads to $(q, q')$ being forced. Any other response by $\bD$ to either of these two moves would lead to her loss in $j - 1$ remaining rounds (by $\bS$ exhibiting the difference between the $c$'s and $d$'s in the \textcolor{red}{$\bJ_{j-1}$} and \textcolor{blue}{$\bJ'_{j-1}$} gadgets), or no longer force the pair $(q, q')$.

Again, this strategy is essentially forced. Any other round $1$ move by $\bS$ can be met by $\bD$ following an automorphism. For instance, if $\bS$ plays on $c_1$, then $\bD$ has a response on $c_9$; if $\bS$ plays on $d_1$, $\bD$ has a response on $d_{13}$; if $\bS$ plays on $d_5$, $\bD$ has a response on $d_{17}$; and if $\bS$ plays in any of the gadgets attached to these vertices, $\bD$ plays on a similar gadget. After a single round, the $c$'s and $d$'s that are the neighbors of $p$ and $p'$ are indistinguishable, and after a further round, the other middle vertices are indistinguishable, and then $\bD$ can exploit the resulting automorphisms easily. It is easy to see that this translates inductively to a winning strategy for $\bD$ in the $j$-round {\ef} game.
\end{proof}

When we eventually play the {\ms} game, we shall use both directions of Lemma \ref{polarity2}, to borrow the critical elements from the {\ef} game analysis and apply them to our setting.

\paragraph*{The Reduction Instance}

Now, consider an arbitrary $\qsat$ instance $\langle \exists x_1\forall x_2\ldots \exists x_{2k - 1}\forall x_{2k}\psi, t\rangle$, where $\psi = C_1 \land \ldots \land C_m$ is the conjunction of $m$ $3$-CNF clauses. (We can assume WLOG that the quantifier signature is alternating, since we can always have ``dummy'' quantified variables without changing the quantifier-free part; this does not affect the set of solutions for the given instance.) We construct the directed graph $\bA$ depicted in Fig.\ \ref{figAPSPACE}, called the \emph{skyscraper}. This structure is formed as follows:
\begin{itemize}
    \item We start by taking $k$ gadgets, consisting of a copy of $\bI_{2k + m - t}$, then a copy of $\bI_{2k-2 + m - t}$, and so on, all the way down to a copy of $\bI_{2 + m - t}$, and stacking them as shown in Fig.\ \ref{figAPSPACE}, with $\bI_{2k + m - t}$ on top, and $\bI_{2 + m - t}$ at the bottom. We identify the $q$ and $q'$ from each of these gadgets with the $p$ and $p'$ respectively from the gadget below it. In the resulting construction, we call the building block $\bI_{2j + m - t}$ the \emph{$j$-th floor of the skyscraper}.
    \item We index the $p, p', q, q'$ vertices in the $j$-th floor by the number of the floor, i.e., $p_j, p'_j, q_j, q'_j$. Apart from that, the $j$-th floor consists of the the upper middle vertices that are the out-neighbors of $p_j$ or $p'_j$. We index those vertices by the superscript $j, U$ (to designate that they are the upper middle vertices in the $j$-th floor). We are especially interested in the vertices $\{c^{j,U}_5, \ldots, c^{j,U}_8\}$, which are the $z$-vertices with \textcolor{blue}{$\bJ'_{2j + m - t - 1}$} gadgets attached to them. We classify them into the following \emph{informal labels}:
    \begin{itemize}
        \item $c^{j,U}_5, c^{j,U}_6$ both get the informal label $T(x_{2k-2j+1})$ (intentionally named, to evoke that playing on this vertex corresponds to setting the variable $x_{2k-2j+1}$ to TRUE),
        %\jon{Perhaps indicate that this will correspond to a truth setting of some sort, and explain what the index $2k-2j+1$ will be.}
        \item $c^{j,U}_7, c^{j,U}_8$ both get the informal label $F(x_{2k-2j+1})$ (to evoke that this corresponds to setting the variable $x_{2k-2j+1}$ to FALSE).
        %\jon{Analogously indicate that this will correspond to a false setting of some sort.}
    \end{itemize}
    The $j$-th floor also consists of a number of lower middle vertices, forming the out-neighbors of the $c^{j, U}_i$'s and the $d^{j, U}_i$'s. We index those vertices by the superscript $j, L, c_i$ or $j, L, d_i$, respectively (to designate that they are the lower middle vertices in the $j$-th floor). We are especially interested in the vertices $\{c^{j,L, c_t}_1, \ldots, c^{j,L, c_t}_4, d^{j,L, c_t}_1, \ldots, d^{j,L, c_t}_4\}$ for $t \in \{5, 6, 7, 8\}$, which are the vertices in the middle of the \textcolor{red}{$\bJ_{2j + m - t - 1}$} gadgets attached to $c_5, \ldots, c_8$. We classify them into the following informal labels, for all $t \in \{5, 6, 7, 8\}$:
    \begin{itemize}
        \item $c^{j,L,c_t}_1$, $c^{j,L,c_t}_2$, $d^{j,L,c_t}_1$, and $d^{j,L,c_t}_2$ get the informal label $T(x_{2k-2j+2})$,
        \item $c^{j,L,c_t}_3$, $c^{j,L,c_t}_4$, $d^{j,L,c_t}_3$, and $d^{j,L,c_t}_4$ get the informal label $F(x_{2k-2j+2})$.
    \end{itemize}
    No other vertex gets an informal label.
    \item Next, we replace $q_1$ and $q'_1$ (in the lowest $\bI_{2 + m - t}$ gadget, i.e., the $1$st floor) with independent sets $S_{q_1}$ and $S_{q'_1}$ of size $m$ and $2m$ respectively. (Recall that $m$ is the number of clauses in the quantifier-free part $\psi$ of the input formula.)
    %\jon{I would here recall that $m$ is the number of clauses in the quantified 3SAT sentence.}
    As before, replacing a vertex preserves the set of incoming and outgoing edges to and from every vertex in the replacing set. Here, $S_{q_1}$ consists of $m$ vertices $\{v_{C_1}, \ldots, v_{C_m}\}$, and $S_{q'_1}$ consists of $2m$ vertices $\{v'_{C_1}, \ldots, v'_{C_m}, \mathsf{null}_1, \ldots, \mathsf{null}_m\}$.
    \item Finally,
    %\jon{A second `Finally' for this paragraph. Perhaps change the prior one to `Next'.}
    for every clause $C_i$ in $\psi$, and for every literal $x_j$ or $\lnot x_j$ appearing in it, we add the edges $(v_{C_i}, v')$ and $(v'_{C_i}, v')$, where $v'$ is a vertex with informal label $T(x_i)$ (if the literal is $x_i$), or with informal label $F(x_i)$ (if the literal is $\lnot x_i$). Note that even after this process, the vertices $\mathsf{null}_1, \ldots, \mathsf{null}_m$ do not have out-neighbors.
\end{itemize}

\input{Images/APSPACE}

Now, let $a = p_k$, and $a' = p'_k$. We can now state and prove our main result.

\begin{restatable}{theorem}{winmslower}\label{thm:winmslower}
The problem $\winms$ is $\PSPACE$-hard.
\end{restatable}

\begin{proof}
We use the construction, $\bA$, and $\qsat$. In particular, given an instance $\langle \varphi, k\rangle$ of $\qsat$, we construct the directed graph $\bA$ and show that:
\begin{align}
    \langle \varphi, t\rangle \in \qsat &\implies \langle{2k + m - t + 2}, (\{\langle \bA ~|~ a\rangle\}, \{\langle \bA ~|~ a'\rangle\})\rangle \in \winms. \label{eq3}\\
    \langle \varphi, t\rangle \notin \qsat &\implies \langle{2k + m - t + 1}, (\{\langle \bA ~|~ a\rangle\}, \{\langle \bA ~|~ a'\rangle\})\rangle \notin \winms. \label{eq4}
\end{align}
We will then use the known inapproximability of $\maxqsat$ (Theorem \ref{qsatapprox}) to obtain our result.

\begin{itemize}
    \item \textbf{Satisfiability Implies Spoiler Wins the $(2k + m - t + 2)$-Round Game:} Suppose that $t$ clauses of $\varphi$ are simultaneously satisfiable. We wish to show that $\bS$ wins the $(2k + m - t + 2)$-round $\ms$ game on $(\{\langle \bA ~|~ a\rangle\}, \{\langle \bA ~|~ a'\rangle\})$. We will again use parallel play to partition the game into multiple parallel sub-games, each corresponding to a different isomorphism class, and then playing them in parallel. The main idea will be that in rounds $1$ through $2k$, $\bS$ and $\bD$ will set values for the variables $x_1, \ldots, x_{2k}$ in order (coincidentally, playing pebble $x_i$ in the round when variable $x_i$ is set), and the rest of the game will involve $\bS$ forcing $\bD$ to pick out clauses that are not satisfied. The round where $\bD$ is unable to do so, therefore, would (roughly) correspond to the point where clauses cannot be simultaneously satisfied any more.
    
    We focus once again first on the top-most floor of the skyscraper, the block $\bI_{2k + m - t}$. Suppose the strategy for $\mathsf{Player}_\exists$ is to set the variable $x_1$ to TRUE. $\bS$ now looks to play on vertices with informal label $T(x_1)$. Therefore, he makes his first round move (with pebble $x_1$) on the left, playing on the upper middle vertex $c^{k, U}_5$ or $c^{k, U}_6$ (since they both have informal label $T(x_1)$). $\bD$ has the ability to make copies for different responses. The only responses preserving isomorphism are when she plays on an outneighbor of $a'$, and all other copies can be discarded by Lemma \ref{lem:discard}. Partition the remaining pebbled structures on the right as:
    \begin{align*}
        \cB_1 &:= \{\langle \bA ~|~ a', c^{k, U}_i\rangle ~:~ i \in \{9, 10, 11, 12\}\} \\
        \cB_2 &:= \{\langle \bA ~|~ a', d^{k, U}_i\rangle ~:~ 9 \leq i \leq 20\}.
    \end{align*}

    $\bS$ now plays on the right (with pebble $x_2$), on these copies. On every pebbled structure in $\langle \bA ~|~ a', c^{k, U}_i\rangle$ in $\cB_1$, he plays on $c^{k, L, c^{k, U}_i}_2$. On every pebbled structure in $\cB_2$, he plays on a distinguishing neighbor of the element with pebble $x_1$. Redefine $\cB_1, \cB_2$ to these new sets (with the round-$2$ pebbles).
    
    Now, $\bD$ responds on the left in all possible ways, but the only responses to consider are when she plays on an out-neighbor of the element with pebble $x_1$ (or on a distinguishing neighbor). Consider each of her possible responses, and partition the left side accordingly. If she responds on a vertex of the form $c^{k, L, c^{k, U}_5}_i$ (which necessarily has an informal label $T(x_2)$ or $F(x_2)$), put the resulting pebbled structure in $\cA_1$. If she responds on a vertex that is a distinguishing neighbor of the element with pebble $x_1$, put the resulting pebbled structure in $\cA_2$. Finally, if she responds on a vertex of the form $d^{k, L, c^{k, U}_5}_i$, put the resulting pebbled structure in $\cA_3$.

    Now, $\bS$ responds on the left in round $3$, with pebble $x_3$. On every pebbled structure in $\cA_2$ (resp.~$\cA_3$), he plays on another distinguishing neighbor of the element with pebble $x_1$ (resp.~$x_2$). Finally, consider $\cA_1$, and note that for every pair of pebbled structures in $\cA_1\times\cB_1$, the pair $(q_k, q'_k)$ is forced. Now, $\bS$ considers each pebbled structure in $\cA_1$ in turn. On each such pebbled structure, he has played pebble $x_1$ on an element with informal label $T(x_1)$ or $F(x_1)$, and $\bD$ has responded with pebble $x_2$ on an element with informal label $T(x_2)$ or $F(x_2)$. These have effectively set the variables $x_1$ and $x_2$ on that structure. Consider the optimal $3$rd move of $\mathsf{Player}_\exists$ in the QBF game with the same setting of variables $x_1$ and $x_2$; whatever this move is, $\bS$ mimics that response by playing appropriately on the upper middle vertex $c_5^{k-1, U}$ or $c_6^{k-1, U}$ (if the optimal move for $\mathsf{Player}_\exists$ is to set $x_3$ to TRUE), or on the upper middle vertex $c_7^{k-1, U}$ or $c_8^{k-1, U}$ (if the optimal move for $\mathsf{Player}_\exists$ is to set $x_3$ to FALSE). $\bS$ does this for all structures in $\cA_1$. Now, consider the responses by $\bD$. The only responses maintaining an isomorphism with the structures in $\cA_2$ are responses of the same form in $\cB_2$. The only responses maintaining an isomorphism with $\cA_1$ or $\cA_3$ form disjoint subsets of $\cB_1$. Rename these subsets $\cB_1$ (maintaining an isomorphism with $\cA_1$) and $\cB_3$ (maintaining an isomorphism with $\cA_3$) respectively.
    
    Consider the sub-games $(\cA_i, \cB_i)$ for $i \in \{1, 2, 3\}$, and note that no pebbled structures from $\cA_i$ and $\cB_j$ form a matching pair, for $i \neq j$. Furthermore, $\bS$ can win the sub-games $(\cA_2, \cB_2)$ and $(\cA_3, \cB_3)$ in the next $2k + 1$ rounds just by counting, as described before in the proof of Theorem \ref{MS-NPhard}; regardless of which side he wants to play in, he can just play on another distinguishing neighbor of the element with pebble $x_1$ (or $x_2$), and run $\bD$ out of moves. (Indeed, note that he might need one additional round, since he eventually needs to play on the left to win the $(\cA_2, \cB_2)$-subgame and on the right to win the $(\cA_2, \cB_2)$-subgame.)
    
    Now, consider the sub-game $(\cA_1, \cB_1)$, where the pair $(q_k, q'_k)$ had been forced, and consider a pebbled structure in $\cB_1$. If $\bD$'s round-$3$ response has not been on an outneighbor of $q'_k$, then $\bS$ can just wait for an additional round (in order to play on the right side), and then play on the in-neighbor of the element with pebble $x_3$ in order to break all matching pairs; the remaining structures can once again be partitioned (as in $\cB_1 \sqcup \cB_2$) into the ones where the response is on an element $c^{k-1, U}_9, \ldots, c^{k-1, U}_{12}$, or an element of the form $d^{k-1, U}_i$. We can now keep going inductively.
    
    $\bS$ can follow this same strategy on these sub-games in parallel. Note that in the same round, on different pebbled structures, he can possibly play on a $T(x_3)$ or $F(x_3)$ vertex \emph{differently}, depending on what $\bD$'s response on that structure has been so far.
    
    In particular, inductively, $\bS$ plays rounds $2i + 1$ and $2i + 2$ on floor $k - i$ of the skyscraper, by playing left and right respectively. His round $(2i + 1)$ move consists of playing on a $c^{k-i, U}_j$-vertex with label $T(x_{k-2i+1})$ or $F(x_{k-2i+1})$, depending on the setting of the variables so far on that pebbled structure. His round $(2i + 2)$ move consists of playing on the right, in the appropriate $c^{k-i, L, c_r}_j$ vertex, and waiting for $\bD$'s response on the left side, which will necessarily be on a vertex with label $T(x_{k-2i+2})$ or $F(x_{k-2i+2})$ (and in some structures, playing on a distinguishing neighbor to run out $\bD$ in the other sub-game).
    
    Note that, after round $(2k + 1)$, $\bS$ is down to a large number of sub-games, but in each of them, he has either won by counting,
    %\jon{can cut the word `already'}
    or $(q_1, q'_1)$ is forced. Here, the additional $+1$ term comes from the extra round he might need in order to run $\bD$ out on the correct side, or in order to certify a response by $\bD$ on a different vertex even if there is a forced pair. At this point, each of the remaining pebbled structures on the left have effectively set the $2k$ variables $x_1$ through $x_{2k}$, with pebble $x_i$ being on an element with informal label either $T(x_i)$ or $F(x_i)$.
    %\jon{I would merge this single sentence paragraph with the next.}
    %\jon{It seems like my prior counting concern is taken care since after $2k$ rounds, $\bS$ plays at least one of the $k-1$ additional round on the left and at least one on the right}
    
    Finally, in the $m - t + 1$ rounds $(2k + 2)$ through $(2k + m - t + 2)$, $\bS$ keeps playing on a different $\mathsf{null}_i$ vertex in $S_{q'_1}$ on all remaining pebbled structures on the right. Consider $\bD$'s responses on the left. She must respond on distinct vertices in $S_{q_1}$ on each pebbled structure on the left. Any such response, say $v_{C_i}$, picks out a $3$-CNF clause $C_i$ in $\psi$. Note that in order to have a chance to stay in the game, $\bD$ must make sure that $v_{C_i}$ does not have an out-neighbor with a pebble (since none of the $\mathsf{null}_i$ vertices have out-neighbors) on that structure. But in the structure $\bA$, the vertex $v_{C_i}$ has an out-neighbor to every literal appearing in it, and so the only way for $\bD$ to have a valid response is to find a vertex $v_{C_i}$, none of whose valid literals has been pebbled. This means that on that pebbled structure, the assignment of variables has missed every literal appearing in $C_i$, and so this assignment keeps the clause $C_i$ false. However, since $\bS$ played optimally on every copy, and since $\mathsf{Player}_\exists$ has a strategy to falsify at most $m - t$ clauses, this means that at most $m - t$ of $\bD$'s responses can maintain matching pairs with the other side. By round $(2k + m - t + 2)$,
    %\jon{I would replace with `by round...' since the matching pairs may have all been broken sooner on an easy instance.}
    all matching pairs are broken, giving $\bS$ the win. 
    %\jon{All matching pairs could be broken sooner if this is not an optimally difficult instance.} \rik{I suppose by Spoiler cleverly keeping track of left/right, he can save one move; but since we don't need that for the approximation to go through, and I'm tired of bookkeeping that edge case, I'm inclined to leave this as it is, unless you strongly object.}
    
    \item \textbf{Non-Satisfiability Implies Duplicator Wins the $(2k + m - t + 1)$-Round Game:} Consider the {\ef} game on the same instance, and suppose that $\varphi$ does not have $t$ clauses simultaneously satisfiable. We know then by the arguments from Theorem 6 in \cite{Pezzoli98} that $\bS$ cannot force the pair $(q_1, q'_1)$ in fewer than $2k + 1$ rounds. Therefore, he cannot exhibit the difference between the two sides in $2k + 1 + m - t$ rounds (since his only strategy is to keep playing on the $\mathsf{null}_i$ vertices until $\bD$ has no response). Note that since $t$ clauses are not simultaneously satisfiable, $\mathsf{Player}_\forall$ has a strategy that yields at msot $t - 1$ satisfied clauses, and hence at least $m - t + 1$ unsatisfied clauses. The structure at the end of the first $2k$ rounds, therefore, corresponds to this assignment, and so $\bD$ has the ability to respond for $m - t + 1$ additional rounds with vertices $v_{C_i}$ corresponding to unsatisfied clauses $C_i$. Therefore, $\bD$ wins the $(2k + m - t + 1)$-round {\ef} game. Therefore, our result is trivial, since if $\bD$ wins an $\ef$ game on an instance, she certainly wins an $\ms$ game on the same instance.  
    %\jon{This argument does not work, per my email. Also, to appeal to Pezzoli at all in a fashion like this, she would have had to use the exact same gadgets.}
    %\rik{Note that the $m - t$ addition to the skyscraper floors fixes the argument; also, Pezzoli uses the same gadgets in her analysis of the first $2k + 1$ rounds.}
\end{itemize}

\paragraph*{Hardness Using Inapproximability}

Now, consider the following algorithm (Algorithm\ \ref{protocol2}) that computes a $(1 - \varepsilon)$-approximation to $\maxqsat$ given a quantified $3$-CNF formula $\varphi$ (for arbitrarily small $\varepsilon > 0$), using successive calls to $\winms$.
%\jon{Since the description you give in words does not include the For loop of Algorithm 2, perhaps rephrase to ``Now consider the pseudo-code in Algorithm 2 that....''}
%\rik{Don't understand -- the algorithm is fine as stated, right?}
%\jon{Good catch that $t$ has to be iterated downward! I think the game has to be run on a size $2k + m - t + 2$ WIN(MS) instance (not size $2k + m - t + 1$) and also, that no approximation is needed, since $\bD$ wins the one smaller instance. I am surprised that $\bD$ wins the one smaller instance and not the two smaller instance, however, because of the need to run out on the appropriate side.}

\begin{lstlisting}[caption={Computing a $(1 - \varepsilon)$-approximation to $\maxqsat$, for arbitrarily small $\varepsilon > 0$.},label=protocol2,captionpos=t,float, escapeinside={(*}{*)}]
On input (*$\varphi$*): 
    For (*$t = m, \ldots, 1$*):
        Construct structure (*$\bA$*) as described;
        Run (*$\winms$*) on (*$\langle 2k + m - t + 1, (\{\langle\bA ~|~ a\rangle\}, \{\langle\bA ~|~ a'\rangle\}\rangle$*);
        if (*$\winms$*) outputs YES:
            output (*$t$*);
\end{lstlisting}

%\jon{I am not following the numbers in this algorithm and the connection to the text. In Algorithm 2, why not run WIN(MS) on $<2k+m-t+2, ({<A | a>}, {<A | a'>})>$ and if WIN(MS) outputs YES, output YES for this $t$. It seems like the given output, $t - 2k - m -3$, will always be negative and I don't understand its connection to anything.}
This algorithm outputs the minimum $t$ such that $\bS$ wins the $(2k + m  - t + 1)$-round game on $(\{\langle \bA ~|~ a\rangle\}, \{\langle \bA ~|~ a'\rangle\})$. To see that this works as claimed, note that, if $t'$ is the maximum number of simultaneously satisfiable clauses, then on the one hand, $\langle \varphi, t'+1\rangle \notin \qsat$, and so, by equation \eqref{eq4}, $\bD$ wins the game instance with $(2k + m - t')$ rounds or fewer; on the other hand, $\langle \varphi, t'\rangle \in \qsat$, and so, by equation \eqref{eq3}, $\bS$ wins the game instance with $(2k + m - t' + 2)$ rounds or more. Algorithm \ref{protocol2} therefore outputs a number between $t'-1$ and $t'$. The approximation ratio, therefore, is $1 - 1/t'$, which is better than any fixed constant when $t'$ is large enough. 

However, since by Theorem \ref{qsatapprox}, the problem $\maxqsat$ is $\PSPACE$-hard to approximate to some multiplicative constant, it follows that $\winms$ is $\PSPACE$-hard as well.
\end{proof}

We note that this same structure can be easily adapted for the $\PSPACE$-hardness for {\ef} games as well. This shows that the hardness result holds even on directed graphs, resolving (in the positive) a question from \cite{Pezzoli98}. We omit the straightforward details, but state this as a corollary.

\begin{corollary}
The problem $\winef$ is $\PSPACE$-hard, even on directed graphs.
\end{corollary}

%% file: Images/Ij.tex
\begin{figure*}[ht]
\begin{center}
\tikzset{Red/.style={circle,draw=red, inner sep=1pt, minimum size=1cc, line width=2pt}}
\tikzset{RedOuter/.style={circle,draw=red, inner sep=1pt, minimum size=1.2cc, line width=2pt}}
\tikzset{BlueOuter/.style={circle,draw=blue, inner sep=1pt, minimum size=1.25cc, line width=2pt}}
\tikzset{Blue/.style={circle,draw=blue, inner sep=1pt, minimum size=1cc, line width=2pt}}
\tikzset{Green/.style={circle,draw=dg, inner sep=1pt, minimum size=1cc, line width=2pt}}
\tikzset{Magenta/.style={circle,draw=magenta, inner sep=1pt, minimum size=1cc, line width=2pt}}
\tikzset{Indigo/.style={circle,draw=indigo, inner sep=1pt, minimum size=1cc, line width=2pt}}
\tikzset{Black/.style={circle,draw=black, inner sep=1pt, minimum size=1cc, line width=1pt}}
\tikzset{Gray/.style={circle,draw=black, inner sep=1pt, minimum size=1cc, line width=1pt, fill=gray!20}}
\tikzset{TreeNode/.style={rectangle,draw=black, inner sep=1pt, minimum size=1.3cc, line width=1pt}}
\tikzset{BigTreeNode/.style={rectangle,draw=black, inner sep=1pt, minimum size=1.8cc, line width=1pt}}
\tikzset{Tleft/.style={isosceles triangle,anchor=west,draw=black, inner sep=1pt, minimum size=1cc,
    line width=1pt,fill=orange}}
\tikzset{TBleft/.style={isosceles triangle,anchor=west,draw=blue, inner sep=1pt, minimum size=1cc,
    line width=2pt,fill=orange}}
\tikzset{Tright/.style={isosceles triangle,anchor=west,draw=black, inner sep=1pt, minimum size=1cc,
    line width=1pt,fill=indigo,rotate=180}}
\tikzset{Trap/.style={trapezium, minimum height=3.2cm,draw}}
\begin{tikzpicture}[scale=.12]
\node [Black] (h1) at (-23.75,8) {{\scriptsize $a^1_1$}};
\node [Black] (h2) at (-26,4) {{\scriptsize $a^1_2$}};
\node [Black] (h3) at (-28.25,0) {{\scriptsize $a^1_3$}};
\node              at (-29.9375,-3) {$\cdot$};
\node              at (-30.5,-4) {$\cdot$};
\node              at (-31.0625,-5) {$\cdot$};
\node [Black] (hk) at (-32.75,-8) {{\scriptsize $a^1_j$}};
\node [Red] (p1) at(-8,12.5)  {{\scriptsize $p$}};
\node [Red] (r1') at(8,12.5)  {{\scriptsize  $p'$}};
\node [Black] (i1) at(-24,0)  {{\scriptsize  $c_1$}};
\node [Black] (i2) at(-17,0)  {{\scriptsize $c_2$}};
\node  at (0,7) {{\large $\bI_j$}};
%\node [Trap] at (0,0) {\rule{2mm}{0mm}};
\node [Black] (m1) at(-10,0)  {{\scriptsize $d_1$}};
\node [Black] (m1') at(-3,0)  {{\scriptsize $d_2$}};
\node [Black] (i3) at(3,0)  {{\scriptsize $c_3$}};
\node [Black] (i4) at(10,0)  {{\scriptsize $c_4$}};
\node [Black] (m3) at(17,0)  {{\scriptsize  $d_3$}};
\node [Black] (m4) at(24,0)  {{\scriptsize $d_4$}};
\node [Blue] (g1) at(-28,-12.5)  {{\scriptsize $q$}};
\node [Blue] (g1') at(-12,-12.5)  {{\scriptsize $q'$}};
\node [Green] (b1) at(12,-12.5)  {{\scriptsize  $r$}};
\node [Green] (b2) at(28,-12.5)  {{\scriptsize  $r'$}};
\node [Black] (j2) at (26,4) {{\scriptsize $b^4_2$}};
\node [Black] (j3) at (28.25,0) {{\scriptsize $b^4_3$}};
\node              at (29.9375,-3) {$\cdot$};
\node              at (30.5,-4) {$\cdot$};
\node              at (31.0625,-5) {$\cdot$};
\node [Black] (jk) at (32.75,-8) {{\scriptsize $b^4_j$}};
\foreach \from/\to in {p1/i1,p1/i2,p1/m1,p1/m1',i1/g1,i1/b1,i2/g1',i2/b2,m1/g1,m1/b2,m1'/g1',m1'/b1,
  r1'/i3,r1'/i4,r1'/m3,r1'/m4,i3/g1,i3/b2,i4/g1',i4/b1,m3/g1,m3/b1,m4/g1',m4/b2,
  h1/i1,h2/i1,h3/i1,hk/i1,j2/m4,j3/m4,jk/m4}
\draw[line width=1pt,color=black,->] (\from) -- (\to);
\end{tikzpicture}
\end{center}
\caption{Gadget $\bI_j$. The  vertices $c_1, \ldots, c_4$ are each  connected to $j$ auxiliary vertices; the vertices  $d_1 \ldots d_4$  are each connected to 
$j-1$ auxiliary vertices. These auxiliary vertices are drawn only for $c_1$ and $d_4$. Such a gadget exists for any $j > 0$.}
\label{figIj}
\end{figure*}
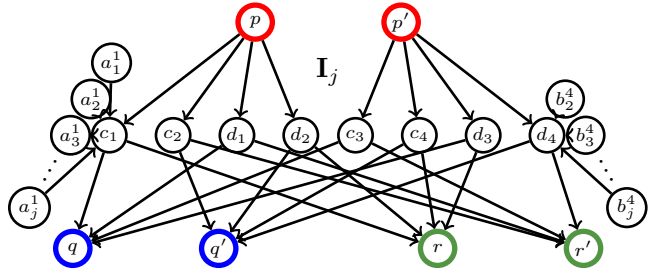

%% file: Images/auxX.tex
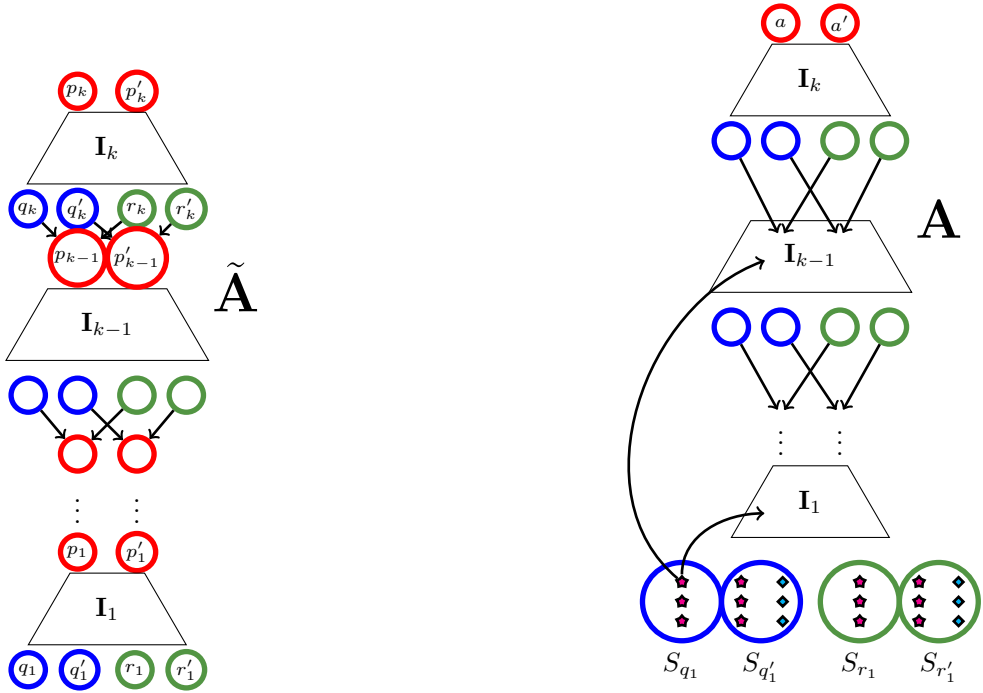
\begin{figure*}[ht]
\begin{center}
\tikzset{Red/.style={circle,draw=red, inner sep=1pt, minimum size=1cc, line width=2pt}}
\tikzset{RedOuter/.style={circle,draw=red, inner sep=1pt, minimum size=1.2cc, line width=2pt}}
\tikzset{BlueOuter/.style={circle,draw=blue, inner sep=1pt, minimum size=1.25cc, line width=2pt}}
\tikzset{Blue/.style={circle,draw=blue, inner sep=1pt, minimum size=1cc, line width=2pt}}
\tikzset{Green/.style={circle,draw=dg, inner sep=1pt, minimum size=1cc, line width=2pt}}
\tikzset{Magenta/.style={circle,draw=magenta, inner sep=1pt, minimum size=1cc, line width=2pt}}
\tikzset{Indigo/.style={circle,draw=indigo, inner sep=1pt, minimum size=1cc, line width=2pt}}
\tikzset{Black/.style={circle,draw=black, inner sep=1pt, minimum size=1cc, line width=1pt}}
\tikzset{DiamondS/.style={diamond,draw=black, fill=cyan, inner sep=1pt, minimum size=0.3cc, line width=1pt}}
\tikzset{StarS/.style={star,draw=black, fill=magenta, inner sep=1pt, minimum size=0.3cc, line width=1pt}}
\tikzset{Gray/.style={circle,draw=black, inner sep=1pt, minimum size=1cc, line width=1pt, fill=gray!20}}
\tikzset{White/.style={circle,draw=black, inner sep=1pt, minimum size=1cc, line width=1pt}}
\tikzset{TreeNode/.style={rectangle,draw=black, inner sep=1pt, minimum size=1.3cc, line width=1pt}}
\tikzset{BigTreeNode/.style={rectangle,draw=black, inner sep=1pt, minimum size=1.8cc, line width=1pt}}
\tikzset{Tleft/.style={isosceles triangle,anchor=west,draw=black, inner sep=1pt, minimum size=1cc,
    line width=1pt,fill=orange}}
\tikzset{TBleft/.style={isosceles triangle,anchor=west,draw=blue, inner sep=1pt, minimum size=1cc,
    line width=2pt,fill=orange}}
\tikzset{Tright/.style={isosceles triangle,anchor=west,draw=black, inner sep=1pt, minimum size=1cc,
    line width=1pt,fill=indigo,rotate=180}}
\tikzset{Trap/.style={trapezium, minimum height=0.95cm,draw}}

\tikzset{Red/.style={circle,draw=red, inner sep=1pt, minimum size=1cc, line width=2pt}}
\tikzset{RedOuter/.style={circle,draw=red, inner sep=1pt, minimum size=1.2cc, line width=2pt}}
\tikzset{BlueOuter/.style={circle,draw=blue, inner sep=1pt, minimum size=1.25cc, line width=2pt}}
\tikzset{Blue/.style={circle,draw=blue, inner sep=1pt, minimum size=1cc, line width=2pt}}
\tikzset{BigBlue/.style={ellipse,draw=blue, inner sep=1pt, minimum size=2.3cc, line width=2pt}}
\tikzset{Green/.style={circle,draw=dg, inner sep=1pt, minimum size=1cc, line width=2pt}}
\tikzset{BigGreen/.style={ellipse,draw=dg, inner sep=1pt, minimum size=2.3cc, line width=2pt}}
\tikzset{Magenta/.style={circle,draw=magenta, inner sep=1pt, minimum size=1cc, line width=2pt}}
\tikzset{Indigo/.style={circle,draw=indigo, inner sep=1pt, minimum size=1cc, line width=2pt}}
\tikzset{Black/.style={circle,draw=black, inner sep=1pt, minimum size=1cc, line width=1pt}}
\tikzset{Gray/.style={circle,draw=black, inner sep=1pt, minimum size=1cc, line width=1pt, fill=gray!20}}
\tikzset{White/.style={circle,draw=black, inner sep=1pt, minimum size=1cc, line width=1pt}}
\tikzset{Invis/.style={circle,draw=none, inner sep=0pt, minimum size=0.3cc, line width=1pt}}
\tikzset{TreeNode/.style={rectangle,draw=black, inner sep=1pt, minimum size=1.3cc, line width=1pt}}
\tikzset{BigTreeNode/.style={rectangle,draw=black, inner sep=1pt, minimum size=1.8cc, line width=1pt}}
\tikzset{Tleft/.style={isosceles triangle,anchor=west,draw=black, inner sep=1pt, minimum size=1cc,
    line width=1pt,fill=orange}}
\tikzset{TBleft/.style={isosceles triangle,anchor=west,draw=blue, inner sep=1pt, minimum size=1cc,
    line width=2pt,fill=orange}}
\tikzset{Tright/.style={isosceles triangle,anchor=west,draw=black, inner sep=1pt, minimum size=1cc,
    line width=1pt,fill=indigo,rotate=180}}
\tikzset{Trap/.style={trapezium, minimum height=0.95cm,draw}}

\begin{tikzpicture}[scale=.13]
\node [Red] (p1) at (-3,60) {{\scriptsize $p_k$}};
\node [Red] (p1') at (3,60) {{\scriptsize $p_k'$}};
\node [Trap] at (0,54.2) {$\bI_k$};
\node [Green] (m1) at (3,48) {{\scriptsize $r_k$}};
\node [Green] (m1') at (8,48) {{\scriptsize $r_k'$}};
\node [Blue] (g1) at (-8,48) {{\scriptsize $q_k$}};
\node [Blue] (g1') at (-3,48) {{\scriptsize $q_k'$}};
\node [Trap] at (0,36.2) {$\bI_{k-1}$};
\node [Red] (p2) at (-3,43) {{\scriptsize $p_{k-1}$}};
\node [Red] (p2') at (3,43) {{\scriptsize $p_{k-1}'$}};
\node [Green] (m2) at (3,29) {};
\node [Green] (m2') at (8,29) {};
\node [Blue] (g2) at (-8,29) {};
\node [Blue] (g2') at (-3,29) {};
\node [Red] (p3) at (-3,23) {};
\node [Red] (p3') at (3,23) {};
\node at (-3,18) {$\vdots$};
\node at (3,18) {$\vdots$};
\node [Red] (pb) at (-3,13) {{\scriptsize $p_1$}};
\node [Red] (pb') at (3,13) {{\scriptsize $p_1'$}};
\node [Trap] at (0,7.2) {$\bI_1$};
\node [Green] (mb) at (2.8,1) {{\scriptsize $r_1$}};
\node [Green] (mb') at (8,1) {{\scriptsize $r_1'$}};
\node [Blue] (gb) at (-8,1) {{\scriptsize $q_1$}};
\node [Blue] (gb') at (-2.8,1) {{\scriptsize $q_1'$}};

\node [Invis] (name) at (13,40) {{\huge $\tilde{\bA}$}};

\foreach \from/\to in {g1/p2,g1'/p2',m1/p2,m1'/p2',g2/p3,g2'/p3',m2/p3,m2'/p3'}
\draw[line width=1pt,color=black,->] (\from) -- (\to);
\end{tikzpicture}
\hspace{1.5in}
\begin{tikzpicture}[scale=.13]
\node [Red] (p1) at (-3,60) {{\scriptsize $a$}};
\node [Red] (p1') at (3,60) {{\scriptsize $a'$}};
\node [Trap] at (0,54.2) {$\bI_k$};
\node [Green] (m1) at (3,48) {};
\node [Green] (m1') at (8,48) {};
\node [Blue] (g1) at (-8,48) {};
\node [Blue] (g1') at (-3,48) {};
\node [Trap] at (0,36.2) {$\bI_{k-1}$};
\node [Invis] (p2) at (-3,38) {};
\node [Invis] (p2') at (3,38) {};
\node [Green] (m2) at (3,29) {};
\node [Green] (m2') at (8,29) {};
\node [Blue] (g2) at (-8,29) {};
\node [Blue] (g2') at (-3,29) {};
\node [Invis] (p3) at (-3,20) {};
\node [Invis] (p3') at (3,20) {};
\node at (-3,18) {$\vdots$};
\node at (3,18) {$\vdots$};
\node [Invis] (pb) at (-3,13) {};
\node [Invis] (pb') at (3,13) {};
\node [Trap] at (0,11.2) {$\bI_1$};
\node [BigGreen] (mb) at (5,1) [label=below:$S_{r_1}$]{};
\node [BigGreen] (mb') at (13,1) [label=below:$S_{r'_1}$]{};
\node [BigBlue] (gb) at (-13,1) [label=below:$S_{q_1}$]{};
\node [BigBlue] (gb') at (-5,1) [label=below:$S_{q'_1}$]{};

\node [Invis] (out1) at (-4,10) {};
\node [Invis] (out2) at (-4,36) {};
\node [Invis] (name) at (13,40) {{\huge $\bA$}};

\node [DiamondS] (v1) at (-2.8,3) {};
\node [DiamondS] (v1) at (-2.8,1) {};
\node [DiamondS] (v1) at (-2.8,-1) {};
\node [StarS] (v1) at (-7,3) {};
\node [StarS] (v1) at (-7,1) {};
\node [StarS] (v1) at (-7,-1) {};
\node [StarS] (v1) at (-13,-1) {};
\node [StarS] (v1) at (-13,1) {};
\node [StarS] (v1) at (-13,3) {};

\node [StarS] (y1) at (5,3) {};
\node [StarS] (y1) at (5,1) {};
\node [StarS] (y1) at (5,-1) {};
\node [StarS] (y1) at (11,3) {};
\node [StarS] (y1) at (11,1) {};
\node [StarS] (y1) at (11,-1) {};
\node [DiamondS] (y1) at (15,-1) {};
\node [DiamondS] (y1) at (15,1) {};
\node [DiamondS] (y1) at (15,3) {};

\foreach \from/\to in {g1/p2,g1'/p2',m1/p2,m1'/p2',g2/p3,g2'/p3',m2/p3,m2'/p3'}
\draw[line width=1pt,color=black,->] (\from) -- (\to);

\draw[line width=1pt,color=black,->] (v1) to[out=90,in=180] (out1);
\draw[line width=1pt,color=black,->] (v1) to[out=140,in=200] (out2);

\end{tikzpicture}
\end{center}
\caption{(Left) The auxiliary structure $\tilde{\bA}$, built using gadgets $\bI_1, \bI_2, \ldots, \bI_k$ stacked together as shown. (Right) The structure $\bA$, formed by modifying $\tilde{\bA}$, first by removing the red intermediate vertex pairs, and then by replacing the lowest $q_1, q'_1, r_1, r'_1$ with independent sets with two types of vertices -- graph vertices (magenta stars) and null vertices (cyan diamonds), and finally, by drawing edges from the graph vertices to their neighbors in the upper gadgets. Only two such edges are shown.}
\label{figX}
\end{figure*}

%% file: Images/MSround1.tex
\begin{figure*}[h!]
    \begin{center}
    \tikzset{Red/.style={circle,draw=red, inner sep=1pt, minimum size=1cc, line width=2pt}}
    \tikzset{RedOuter/.style={circle,draw=red, inner sep=1pt, minimum size=1.2cc, line width=2pt}}
    \tikzset{BlueOuter/.style={circle,draw=blue, inner sep=1pt, minimum size=1.25cc, line width=2pt}}
    \tikzset{Blue/.style={circle,draw=blue, inner sep=1pt, minimum size=1cc, line width=2pt}}
    \tikzset{Green/.style={circle,draw=dg, inner sep=1pt, minimum size=1cc, line width=2pt}}
    \tikzset{Magenta/.style={circle,draw=magenta, inner sep=1pt, minimum size=1cc, line width=2pt}}
    \tikzset{Indigo/.style={circle,draw=indigo, inner sep=1pt, minimum size=1cc, line width=2pt}}
    \tikzset{Black/.style={circle,draw=black, inner sep=1pt, minimum size=1cc, line width=1pt}}
    \tikzset{BlackS/.style={rectangle,draw=black, inner sep=0.5pt, minimum size=0.5cc, line width=0.5pt}}
    \tikzset{BlackT/.style={ellipse,draw=black, inner sep=0.3pt, minimum size=0.5cc, line width=0.5pt}}
    \tikzset{Gray/.style={circle,draw=black, inner sep=1pt, minimum size=1cc, line width=1pt, fill=gray!20}}
    \tikzset{Invis/.style={circle,draw=none, inner sep=1pt, minimum size=1cc, line width=1pt, fill=white}}
    \tikzset{TreeNode/.style={rectangle,draw=black, inner sep=1pt, minimum size=1.3cc, line width=1pt}}
    \tikzset{BigTreeNode/.style={rectangle,draw=black, inner sep=1pt, minimum size=1.8cc, line width=1pt}}
    \tikzset{Tleft/.style={isosceles triangle,anchor=west,draw=black, inner sep=1pt, minimum size=1cc,
        line width=1pt,fill=orange}}
    \tikzset{TBleft/.style={isosceles triangle,anchor=west,draw=blue, inner sep=1pt, minimum size=1cc,
        line width=2pt,fill=orange}}
    \tikzset{Tright/.style={isosceles triangle,anchor=west,draw=black, inner sep=1pt, minimum size=1cc,
        line width=1pt,fill=indigo,rotate=180}}
    %\tikzset{Trap/.style={trapezium, minimum height=3.2cm,draw}}
    
        \begin{tikzpicture}[scale=.07]

\node [Red] (p1) at(-8,12.5+50)  {$a$};
\node [Red] (r1') at(8,12.5+50)  {};
\node [BlackS] (i1) at(-24,0+50)  {\small\bf $x_1$};
\node [BlackS] (i2) at(-17,0+50)  {};
%\node  at (0,7) {{\large $\bI_j$}};
%\node [Trap] at (0,0) {\rule{2mm}{0mm}};
\node [BlackT] (m1) at(-10,0+50)  {};
\node [BlackT] (m1') at(-3,0+50)  {};
\node [BlackS] (i3) at(3,0+50)  {};
\node [BlackS] (i4) at(10,0+50)  {};
\node [BlackT] (m3) at(17,0+50)  {};
\node [BlackT] (m4) at(24,0+50)  {};
\node [Blue] (g1) at(-28,-12.5+50)  {};
\node [Blue] (g1') at(-12,-12.5+50)  {};
\node [Green] (b1) at(12,-12.5+50)  {};
\node [Green] (b2) at(28,-12.5+50)  {};

\foreach \from/\to in {p1/i1,p1/i2,p1/m1,p1/m1',i1/g1,i1/b1,i2/g1',i2/b2,m1/g1,m1/b2,m1'/g1',m1'/b1,
  r1'/i3,r1'/i4,r1'/m3,r1'/m4,i3/g1,i3/b2,i4/g1',i4/b1,m3/g1,m3/b1,m4/g1',m4/b2}
\draw[line width=1pt,color=black,->] (\from) -- (\to);

\node [Invis] (z1) at(0,0)  {};
\node [Invis] (z2) at(0, 100)  {};
\end{tikzpicture}
\hspace{1in}
\begin{tikzpicture}[scale=.07]

\node [Red] (p1) at(-8,12.5)  {};
\node [Red] (r1') at(8,12.5)  {$a'$};
\node [BlackS] (i1) at(-24,0)  {};
\node [BlackS] (i2) at(-17,0)  {};
%\node  at (0,7) {{\large $\bI_j$}};
%\node [Trap] at (0,0) {\rule{2mm}{0mm}};
\node [BlackT] (m1) at(-10,0)  {};
\node [BlackT] (m1') at(-3,0)  {};
\node [BlackS] (i3) at(3,0)  {};
\node [BlackS] (i4) at(10,0)  {};
\node [BlackT] (m3) at(17,0)  {};
\node [BlackT] (m4) at(24,0)  {\small\bf $x_1$};
\node [Blue] (g1) at(-28,-12.5)  {};
\node [Blue] (g1') at(-12,-12.5)  {};
\node [Green] (b1) at(12,-12.5)  {};
\node [Green] (b2) at(28,-12.5)  {};
\foreach \from/\to in {p1/i1,p1/i2,p1/m1,p1/m1',i1/g1,i1/b1,i2/g1',i2/b2,m1/g1,m1/b2,m1'/g1',m1'/b1,
  r1'/i3,r1'/i4,r1'/m3,r1'/m4,i3/g1,i3/b2,i4/g1',i4/b1,m3/g1,m3/b1,m4/g1',m4/b2}
\draw[line width=1pt,color=black,->] (\from) -- (\to);

\node [Red] (p1) at(-8,12.5+33)  {};
\node [Red] (r1') at(8,12.5+33)  {$a'$};
\node [BlackS] (i1) at(-24,0+33)  {};
\node [BlackS] (i2) at(-17,0+33)  {};
%\node  at (0,7) {{\large $\bI_j$}};
%\node [Trap] at (0,0) {\rule{2mm}{0mm}};
\node [BlackT] (m1) at(-10,0+33)  {};
\node [BlackT] (m1') at(-3,0+33)  {};
\node [BlackS] (i3) at(3,0+33)  {};
\node [BlackS] (i4) at(10,0+33)  {};
\node [BlackT] (m3) at(17,0+33)  {\small\bf $x_1$};
\node [BlackT] (m4) at(24,0+33)  {};
\node [Blue] (g1) at(-28,-12.5+33)  {};
\node [Blue] (g1') at(-12,-12.5+33)  {};
\node [Green] (b1) at(12,-12.5+33)  {};
\node [Green] (b2) at(28,-12.5+33)  {};
\foreach \from/\to in {p1/i1,p1/i2,p1/m1,p1/m1',i1/g1,i1/b1,i2/g1',i2/b2,m1/g1,m1/b2,m1'/g1',m1'/b1,
  r1'/i3,r1'/i4,r1'/m3,r1'/m4,i3/g1,i3/b2,i4/g1',i4/b1,m3/g1,m3/b1,m4/g1',m4/b2}
\draw[line width=1pt,color=black,->] (\from) -- (\to);

\node [Red] (p1) at(-8,12.5+67)  {};
\node [Red] (r1') at(8,12.5+67)  {$a'$};
\node [BlackS] (i1) at(-24,0+67)  {};
\node [BlackS] (i2) at(-17,0+67)  {};
%\node  at (0,7) {{\large $\bI_j$}};
%\node [Trap] at (0,0) {\rule{2mm}{0mm}};
\node [BlackT] (m1) at(-10,0+67)  {};
\node [BlackT] (m1') at(-3,0+67)  {};
\node [BlackS] (i3) at(3,0+67)  {};
\node [BlackS] (i4) at(10,0+67)  {\small\bf $x_1$};
\node [BlackT] (m3) at(17,0+67)  {};
\node [BlackT] (m4) at(24,0+67)  {};
\node [Blue] (g1) at(-28,-12.5+67)  {};
\node [Blue] (g1') at(-12,-12.5+67)  {};
\node [Green] (b1) at(12,-12.5+67)  {};
\node [Green] (b2) at(28,-12.5+67)  {};
\foreach \from/\to in {p1/i1,p1/i2,p1/m1,p1/m1',i1/g1,i1/b1,i2/g1',i2/b2,m1/g1,m1/b2,m1'/g1',m1'/b1,
  r1'/i3,r1'/i4,r1'/m3,r1'/m4,i3/g1,i3/b2,i4/g1',i4/b1,m3/g1,m3/b1,m4/g1',m4/b2}
\draw[line width=1pt,color=black,->] (\from) -- (\to);

\node [Red] (p1) at(-8,12.5+100)  {};
\node [Red] (r1') at(8,12.5+100)  {$a'$};
\node [BlackS] (i1) at(-24,0+100)  {};
\node [BlackS] (i2) at(-17,0+100)  {};
%\node  at (0,7) {{\large $\bI_j$}};
%\node [Trap] at (0,0) {\rule{2mm}{0mm}};
\node [BlackT] (m1) at(-10,0+100)  {};
\node [BlackT] (m1') at(-3,0+100)  {};
\node [BlackS] (i3) at(3,0+100)  {\small\bf $x_1$};
\node [BlackS] (i4) at(10,0+100)  {};
\node [BlackT] (m3) at(17,0+100)  {};
\node [BlackT] (m4) at(24,0+100)  {};
\node [Blue] (g1) at(-28,-12.5+100)  {};
\node [Blue] (g1') at(-12,-12.5+100)  {};
\node [Green] (b1) at(12,-12.5+100)  {};
\node [Green] (b2) at(28,-12.5+100)  {};
\foreach \from/\to in {p1/i1,p1/i2,p1/m1,p1/m1',i1/g1,i1/b1,i2/g1',i2/b2,m1/g1,m1/b2,m1'/g1',m1'/b1,
  r1'/i3,r1'/i4,r1'/m3,r1'/m4,i3/g1,i3/b2,i4/g1',i4/b1,m3/g1,m3/b1,m4/g1',m4/b2}
\draw[line width=1pt,color=black,->] (\from) -- (\to);
\end{tikzpicture}
\end{center}
\caption{State of the game after $\bD$'s round $1$ response, and after applying Lemma \ref{lem:discard}. $\bS$ has played pebble $x_1$ on the left, and the four relevant responses by $\bD$ are shown on the right. The top structure goes into $\cB_1$, the second goes into $\cB_2$, and the bottom two go into $\cB_3$.}
\label{figRound1}
\end{figure*}

%% file: Images/MSround2.tex
\begin{figure*}[h!]
    \begin{center}
    \tikzset{Red/.style={circle,draw=red, inner sep=1pt, minimum size=1cc, line width=2pt}}
    \tikzset{RedOuter/.style={circle,draw=red, inner sep=1pt, minimum size=1.2cc, line width=2pt}}
    \tikzset{BlueOuter/.style={circle,draw=blue, inner sep=1pt, minimum size=1.25cc, line width=2pt}}
    \tikzset{Blue/.style={circle,draw=blue, inner sep=1pt, minimum size=1cc, line width=2pt}}
    \tikzset{Green/.style={circle,draw=dg, inner sep=1pt, minimum size=1cc, line width=2pt}}
    \tikzset{Magenta/.style={circle,draw=magenta, inner sep=1pt, minimum size=1cc, line width=2pt}}
    \tikzset{Indigo/.style={circle,draw=indigo, inner sep=1pt, minimum size=1cc, line width=2pt}}
    \tikzset{Black/.style={circle,draw=black, inner sep=1pt, minimum size=1cc, line width=1pt}}
    \tikzset{BlackS/.style={rectangle,draw=black, inner sep=0.5pt, minimum size=0.5cc, line width=0.5pt}}
    \tikzset{BlackT/.style={ellipse,draw=black, inner sep=0.3pt, minimum size=0.5cc, line width=0.5pt}}
    \tikzset{Gray/.style={circle,draw=black, inner sep=1pt, minimum size=1cc, line width=1pt, fill=gray!20}}
    \tikzset{Invis/.style={circle,draw=none, inner sep=1pt, minimum size=1cc, line width=1pt, fill=white}}
    \tikzset{TreeNode/.style={rectangle,draw=black, inner sep=1pt, minimum size=1.3cc, line width=1pt}}
    \tikzset{BigTreeNode/.style={rectangle,draw=black, inner sep=1pt, minimum size=1.8cc, line width=1pt}}
    \tikzset{Tleft/.style={isosceles triangle,anchor=west,draw=black, inner sep=1pt, minimum size=1cc,
        line width=1pt,fill=orange}}
    \tikzset{TBleft/.style={isosceles triangle,anchor=west,draw=blue, inner sep=1pt, minimum size=1cc,
        line width=2pt,fill=orange}}
    \tikzset{Tright/.style={isosceles triangle,anchor=west,draw=black, inner sep=1pt, minimum size=1cc,
        line width=1pt,fill=indigo,rotate=180}}
    %\tikzset{Trap/.style={trapezium, minimum height=3.2cm,draw}}
    
    \begin{tikzpicture}[scale=.07]

\node [Red] (p1) at(-8,12.5+25)  {};
\node [Red] (r1') at(8,12.5+25)  {$a'$};
\node [BlackS] (i1) at(-24,0+25)  {\small\bf $x_1$};
\node [BlackS] (i2) at(-17,0+25)  {};
%\node  at (0,7) {{\large $\bI_j$}};
%\node [Trap] at (0,0) {\rule{2mm}{0mm}};
\node [BlackT] (m1) at(-10,0+25)  {};
\node [BlackT] (m1') at(-3,0+25)  {};
\node [BlackS] (i3) at(3,0+25)  {};
\node [BlackS] (i4) at(10,0+25)  {};
\node [BlackT] (m3) at(17,0+25)  {};
\node [BlackT] (m4) at(24,0+25)  {};
\node [Blue] (g1) at(-28,-12.5+25)  {};
\node [Blue] (g1') at(-12,-12.5+25)  {};
\node [Green] (b1) at(12,-12.5+25)  {};
\node [Green] (b2) at(28,-12.5+25)  {};
\node [Invis] (z1) at(-24-7,0+25+7)  {\small\bf $x_2$};
\foreach \from/\to in {p1/i1,p1/i2,p1/m1,p1/m1',i1/g1,i1/b1,i2/g1',i2/b2,m1/g1,m1/b2,m1'/g1',m1'/b1,
  r1'/i3,r1'/i4,r1'/m3,r1'/m4,i3/g1,i3/b2,i4/g1',i4/b1,m3/g1,m3/b1,m4/g1',m4/b2}
\draw[line width=1pt,color=black,->] (\from) -- (\to);
\draw[line width=1pt,color=black,->] (z1) -- (i1);

\node [Red] (p1) at(-8,12.5+60)  {};
\node [Red] (r1') at(8,12.5+60)  {$a'$};
\node [BlackS] (i1) at(-24,0+60)  {\small\bf $x_1$};
\node [BlackS] (i2) at(-17,0+60)  {};
%\node  at (0,7) {{\large $\bI_j$}};
%\node [Trap] at (0,0) {\rule{2mm}{0mm}};
\node [BlackT] (m1) at(-10,0+60)  {};
\node [BlackT] (m1') at(-3,0+60)  {};
\node [BlackS] (i3) at(3,0+60)  {};
\node [BlackS] (i4) at(10,0+60)  {};
\node [BlackT] (m3) at(17,0+60)  {};
\node [BlackT] (m4) at(24,0+60)  {};
\node [Blue] (g1) at(-28,-12.5+60)  {\small\bf $x_2$};
\node [Blue] (g1') at(-12,-12.5+60)  {};
\node [Green] (b1) at(12,-12.5+60)  {};
\node [Green] (b2) at(28,-12.5+60)  {};
\foreach \from/\to in {p1/i1,p1/i2,p1/m1,p1/m1',i1/g1,i1/b1,i2/g1',i2/b2,m1/g1,m1/b2,m1'/g1',m1'/b1,
  r1'/i3,r1'/i4,r1'/m3,r1'/m4,i3/g1,i3/b2,i4/g1',i4/b1,m3/g1,m3/b1,m4/g1',m4/b2}
\draw[line width=1pt,color=black,->] (\from) -- (\to);

\node [Red] (p1) at(-8,12.5+95)  {};
\node [Red] (r1') at(8,12.5+95)  {$a'$};
\node [BlackS] (i1) at(-24,0+95)  {\small\bf $x_1$};
\node [BlackS] (i2) at(-17,0+95)  {};
%\node  at (0,7) {{\large $\bI_j$}};
%\node [Trap] at (0,0) {\rule{2mm}{0mm}};
\node [BlackT] (m1) at(-10,0+95)  {};
\node [BlackT] (m1') at(-3,0+95)  {};
\node [BlackS] (i3) at(3,0+95)  {};
\node [BlackS] (i4) at(10,0+95)  {};
\node [BlackT] (m3) at(17,0+95)  {};
\node [BlackT] (m4) at(24,0+95)  {};
\node [Blue] (g1) at(-28,-12.5+95)  {};
\node [Blue] (g1') at(-12,-12.5+95)  {};
\node [Green] (b1) at(12,-12.5+95)  {\small\bf $x_2$};
\node [Green] (b2) at(28,-12.5+95)  {};
\foreach \from/\to in {p1/i1,p1/i2,p1/m1,p1/m1',i1/g1,i1/b1,i2/g1',i2/b2,m1/g1,m1/b2,m1'/g1',m1'/b1,
  r1'/i3,r1'/i4,r1'/m3,r1'/m4,i3/g1,i3/b2,i4/g1',i4/b1,m3/g1,m3/b1,m4/g1',m4/b2}
\draw[line width=1pt,color=black,->] (\from) -- (\to);

\node [Invis] (z1) at(0,0)  {};
\node [Invis] (z2) at(0,100)  {};
\end{tikzpicture}
\hspace{1in}
\begin{tikzpicture}[scale=.07]

\node [Red] (p1) at(-8,12.5)  {};
\node [Red] (r1') at(8,12.5)  {$a'$};
\node [BlackS] (i1) at(-24,0)  {};
\node [BlackS] (i2) at(-17,0)  {};
%\node  at (0,7) {{\large $\bI_j$}};
%\node [Trap] at (0,0) {\rule{2mm}{0mm}};
\node [BlackT] (m1) at(-10,0)  {};
\node [BlackT] (m1') at(-3,0)  {};
\node [BlackS] (i3) at(3,0)  {};
\node [BlackS] (i4) at(10,0)  {};
\node [BlackT] (m3) at(17,0)  {};
\node [BlackT] (m4) at(24,0)  {\small\bf $x_1$};
\node [Blue] (g1) at(-28,-12.5)  {};
\node [Blue] (g1') at(-12,-12.5)  {};
\node [Green] (b1) at(12,-12.5)  {};
\node [Green] (b2) at(28,-12.5)  {};
\node [Invis] (z1) at(24+7,0+7)  {\small\bf $x_2$};

\foreach \from/\to in {p1/i1,p1/i2,p1/m1,p1/m1',i1/g1,i1/b1,i2/g1',i2/b2,m1/g1,m1/b2,m1'/g1',m1'/b1, z1/m4,
  r1'/i3,r1'/i4,r1'/m3,r1'/m4,i3/g1,i3/b2,i4/g1',i4/b1,m3/g1,m3/b1,m4/g1',m4/b2}
\draw[line width=1pt,color=black,->] (\from) -- (\to);

\node [Red] (p1) at(-8,12.5+33)  {};
\node [Red] (r1') at(8,12.5+33)  {$a'$};
\node [BlackS] (i1) at(-24,0+33)  {};
\node [BlackS] (i2) at(-17,0+33)  {};
%\node  at (0,7) {{\large $\bI_j$}};
%\node [Trap] at (0,0) {\rule{2mm}{0mm}};
\node [BlackT] (m1) at(-10,0+33)  {};
\node [BlackT] (m1') at(-3,0+33)  {};
\node [BlackS] (i3) at(3,0+33)  {};
\node [BlackS] (i4) at(10,0+33)  {};
\node [BlackT] (m3) at(17,0+33)  {\small\bf $x_1$};
\node [BlackT] (m4) at(24,0+33)  {};
\node [Blue] (g1) at(-28,-12.5+33)  {};
\node [Blue] (g1') at(-12,-12.5+33)  {};
\node [Green] (b1) at(12,-12.5+33)  {};
\node [Green] (b2) at(28,-12.5+33)  {};
\node [Invis] (z1) at(17+7,0+33+7)  {\small\bf $x_2$};
\foreach \from/\to in {p1/i1,p1/i2,p1/m1,p1/m1',i1/g1,i1/b1,i2/g1',i2/b2,m1/g1,m1/b2,m1'/g1',m1'/b1, z1/m3,
  r1'/i3,r1'/i4,r1'/m3,r1'/m4,i3/g1,i3/b2,i4/g1',i4/b1,m3/g1,m3/b1,m4/g1',m4/b2}
\draw[line width=1pt,color=black,->] (\from) -- (\to);

\node [Red] (p1) at(-8,12.5+67)  {};
\node [Red] (r1') at(8,12.5+67)  {$a'$};
\node [BlackS] (i1) at(-24,0+67)  {};
\node [BlackS] (i2) at(-17,0+67)  {};
%\node  at (0,7) {{\large $\bI_j$}};
%\node [Trap] at (0,0) {\rule{2mm}{0mm}};
\node [BlackT] (m1) at(-10,0+67)  {};
\node [BlackT] (m1') at(-3,0+67)  {};
\node [BlackS] (i3) at(3,0+67)  {};
\node [BlackS] (i4) at(10,0+67)  {\small\bf $x_1$};
\node [BlackT] (m3) at(17,0+67)  {};
\node [BlackT] (m4) at(24,0+67)  {};
\node [Blue] (g1) at(-28,-12.5+67)  {};
\node [Blue] (g1') at(-12,-12.5+67)  {\small\bf $x_2$};
\node [Green] (b1) at(12,-12.5+67)  {};
\node [Green] (b2) at(28,-12.5+67)  {};
\foreach \from/\to in {p1/i1,p1/i2,p1/m1,p1/m1',i1/g1,i1/b1,i2/g1',i2/b2,m1/g1,m1/b2,m1'/g1',m1'/b1,
  r1'/i3,r1'/i4,r1'/m3,r1'/m4,i3/g1,i3/b2,i4/g1',i4/b1,m3/g1,m3/b1,m4/g1',m4/b2}
\draw[line width=1pt,color=black,->] (\from) -- (\to);

\node [Red] (p1) at(-8,12.5+100)  {};
\node [Red] (r1') at(8,12.5+100)  {$a'$};
\node [BlackS] (i1) at(-24,0+100)  {};
\node [BlackS] (i2) at(-17,0+100)  {};
%\node  at (0,7) {{\large $\bI_j$}};
%\node [Trap] at (0,0) {\rule{2mm}{0mm}};
\node [BlackT] (m1) at(-10,0+100)  {};
\node [BlackT] (m1') at(-3,0+100)  {};
\node [BlackS] (i3) at(3,0+100)  {\small\bf $x_1$};
\node [BlackS] (i4) at(10,0+100)  {};
\node [BlackT] (m3) at(17,0+100)  {};
\node [BlackT] (m4) at(24,0+100)  {};
\node [Blue] (g1) at(-28,-12.5+100)  {};
\node [Blue] (g1') at(-12,-12.5+100)  {};
\node [Green] (b1) at(12,-12.5+100)  {};
\node [Green] (b2) at(28,-12.5+100)  {\small\bf $x_2$};
\foreach \from/\to in {p1/i1,p1/i2,p1/m1,p1/m1',i1/g1,i1/b1,i2/g1',i2/b2,m1/g1,m1/b2,m1'/g1',m1'/b1,
  r1'/i3,r1'/i4,r1'/m3,r1'/m4,i3/g1,i3/b2,i4/g1',i4/b1,m3/g1,m3/b1,m4/g1',m4/b2}
\draw[line width=1pt,color=black,->] (\from) -- (\to);
\end{tikzpicture}
\end{center}
\caption{State of the game after $\bD$'s round $2$ response, and after applying Lemma \ref{lem:discard}. $\bS$ has played pebble $x_2$ on the right, and the three relevant responses by $\bD$ are shown on the left. The top structure goes into $\cA_1$, the second goes into $\cA_2$, and the third into $\cA_3$.}
\label{figRound2}
\end{figure*}

%% file: Images/Jj.tex
\begin{figure*}[ht]
    \begin{center}
    \tikzset{Red/.style={circle,draw=red, inner sep=1pt, minimum size=1cc, line width=2pt}}
    \tikzset{RedOuter/.style={circle,draw=red, inner sep=1pt, minimum size=1.2cc, line width=2pt}}
    \tikzset{BlueOuter/.style={circle,draw=blue, inner sep=1pt, minimum size=1.25cc, line width=2pt}}
    \tikzset{Blue/.style={circle,draw=blue, inner sep=1pt, minimum size=1cc, line width=2pt}}
    \tikzset{Green/.style={circle,draw=dg, inner sep=1pt, minimum size=1cc, line width=2pt}}
    \tikzset{Magenta/.style={circle,draw=magenta, inner sep=1pt, minimum size=1cc, line width=2pt}}
    \tikzset{Indigo/.style={circle,draw=indigo, inner sep=1pt, minimum size=1cc, line width=2pt}}
    \tikzset{Black/.style={circle,draw=black, inner sep=1pt, minimum size=1cc, line width=1pt}}
    \tikzset{BlackS/.style={rectangle,draw=black, inner sep=0.5pt, minimum size=0.5cc, line width=0.5pt}}
    \tikzset{BlackT/.style={ellipse,draw=black, inner sep=0.3pt, minimum size=0.5cc, line width=0.5pt}}
    \tikzset{Gray/.style={circle,draw=black, inner sep=1pt, minimum size=1cc, line width=1pt, fill=gray!20}}
    \tikzset{Invis/.style={circle,draw=none, inner sep=1pt, minimum size=1cc, line width=1pt, fill=white}}
    \tikzset{TreeNode/.style={rectangle,draw=black, inner sep=1pt, minimum size=1.3cc, line width=1pt}}
    \tikzset{BigTreeNode/.style={rectangle,draw=black, inner sep=1pt, minimum size=1.8cc, line width=1pt}}
    \tikzset{Tleft/.style={isosceles triangle,anchor=west,draw=black, inner sep=1pt, minimum size=1cc,
        line width=1pt,fill=orange}}
    \tikzset{TBleft/.style={isosceles triangle,anchor=west,draw=blue, inner sep=1pt, minimum size=1cc,
        line width=2pt,fill=orange}}
    \tikzset{Tright/.style={isosceles triangle,anchor=west,draw=black, inner sep=1pt, minimum size=1cc,
        line width=1pt,fill=indigo,rotate=180}}
    %\tikzset{Trap/.style={trapezium, minimum height=3.2cm,draw}}

\begin{tikzpicture}[scale=.08]
\node [Invis]             at (-29.9375-2,-3) {\tiny $\vdots$};
\node [Black] (a11) at (-23.75-2,8+4) {{\scriptsize $a^1_1$}};
\node [Black] (a12) at (-26-4,4+2) {{\scriptsize $a^1_2$}};
\node [Black] (a13) at (-28.25-5,0) {{\scriptsize $a^1_3$}};
\node [Black] (a1j) at (-32.75+2,-10) {{\scriptsize $a^1_j$}};
\node [Black] (z) at(0,12.5)  {{\scriptsize $z$}};
\node [Black] (c1) at(-24,0)  {{\scriptsize  $c_1$}};
\node [Black] (c2) at(-17,0)  {{\scriptsize $c_2$}};
\node  at (0,-12) {{\large $\bJ_j$}};
%\node [Trap] at (0,0) {\rule{2mm}{0mm}};
\node [Black] (c3) at(-10,0)  {{\scriptsize $c_3$}};
\node [Black] (c4) at(-3,0)  {{\scriptsize $c_4$}};
\node [Black] (d1) at(3,0)  {{\scriptsize $d_1$}};
\node [Black] (d2) at(10,0)  {{\scriptsize $d_2$}};
\node [Black] (d3) at(17,0)  {{\scriptsize  $d_3$}};
\node [Black] (d4) at(24,0)  {{\scriptsize $d_4$}};
\node [Black] (q) at(-12,-12.5)  {{\scriptsize $q$}};
\node [Black] (q') at(12,-12.5)  {{\scriptsize $q'$}};
%\node [Green] (b1) at(12,-12.5)  {{\scriptsize  $r$}};
%\node [Green] (b2) at(28,-12.5)  {{\scriptsize  $r'$}};
\node [Black] (b42) at (26+4,4+2) {{\scriptsize $b^4_2$}};
\node [Black] (b43) at (28.25+5,0) {{\scriptsize $b^4_3$}};
\node              at (31.0625+2,-5) {\tiny $\vdots$};
\node [Black] (b4j) at (32.75-2,-10) {{\scriptsize $b^4_j$}};
\foreach \from/\to in {z/c1, z/c2, z/c3, z/c4, z/d1, z/d2, z/d3, z/d4,
c1/q, c2/q, c3/q, c4/q,
d1/q', d2/q', d3/q', d4/q',
a11/c1, a12/c1, a13/c1, a1j/c1,
b42/d4, b43/d4, b4j/d4}
\draw[line width=1pt,color=black,->] (\from) -- (\to);
\end{tikzpicture}
\hspace{0.2in}
\begin{tikzpicture}[scale=.08]
\node [Invis]             at (-29.9375-2,-3) {\tiny $\vdots$};
\node [Black] (a11) at (-23.75-2,8+4) {{\scriptsize $a^1_1$}};
\node [Black] (a12) at (-26-4,4+2) {{\scriptsize $a^1_2$}};
\node [Black] (a13) at (-28.25-5,0) {{\scriptsize $a^1_3$}};
\node [Black] (a1j) at (-32.75+2,-10) {{\scriptsize $a^1_j$}};
\node [Black] (z) at(0,12.5)  {{\scriptsize $z$}};
\node [Black] (c1) at(-24,0)  {{\scriptsize  $c_1$}};
\node [Black] (c2) at(-17,0)  {{\scriptsize $c_2$}};
\node  at (0,-12) {{\large $\bJ'_j$}};
%\node [Trap] at (0,0) {\rule{2mm}{0mm}};
\node [Black] (c3) at(-10,0)  {{\scriptsize $c_3$}};
\node [Black] (c4) at(-3,0)  {{\scriptsize $c_4$}};
\node [Black] (d1) at(3,0)  {{\scriptsize $d_1$}};
\node [Black] (d2) at(10,0)  {{\scriptsize $d_2$}};
\node [Black] (d3) at(17,0)  {{\scriptsize  $d_3$}};
\node [Black] (d4) at(24,0)  {{\scriptsize $d_4$}};
\node [Black] (q) at(-12,-12.5)  {{\scriptsize $q$}};
\node [Black] (q') at(12,-12.5)  {{\scriptsize $q'$}};
%\node [Green] (b1) at(12,-12.5)  {{\scriptsize  $r$}};
%\node [Green] (b2) at(28,-12.5)  {{\scriptsize  $r'$}};
\node [Black] (b42) at (26+4,4+2) {{\scriptsize $b^4_2$}};
\node [Black] (b43) at (28.25+5,0) {{\scriptsize $b^4_3$}};
\node              at (31.0625+2,-5) {\tiny $\vdots$};
\node [Black] (b4j) at (32.75-2,-10) {{\scriptsize $b^4_j$}};
\foreach \from/\to in {z/c1, z/c2, z/c3, z/c4, z/d1, z/d2, z/d3, z/d4,
c1/q, c3/q, d1/q, d3/q,
c2/q', c4/q', d2/q', d4/q',
a11/c1, a12/c1, a13/c1, a1j/c1,
b42/d4, b43/d4, b4j/d4}
\draw[line width=1pt,color=black,->] (\from) -- (\to);
\end{tikzpicture}
\end{center}
\caption{The building blocks $\bJ_j$ and $\bJ'_j$. Copies of these will be connected to form the main building block $\bI_j$ for the $\PSPACE$ construction. To avoid clunky diagrams, we shall designate $\bJ_j$ blocks in \textcolor{red}{red}, and $\bJ'_j$ blocks in \textcolor{blue}{blue}.
%\jon{It seems from the bigger construction that $p$ is never the top node in either of these two gadgets. In fact, the top node in each gadget is always a node that is labeled by $c_i$ for $1 \leq i \leq 12$ or $d_j$ for $1 \leq j \leq 20$ in the big gadget. It is confusing to me that $c_i$ and $d_j$ are used to label these nodes right below $p$ in the big gadget and also the nodes within the respective $\bJ_j$ and $\bJ'_j$.}
}
\label{figJj}
\end{figure*}
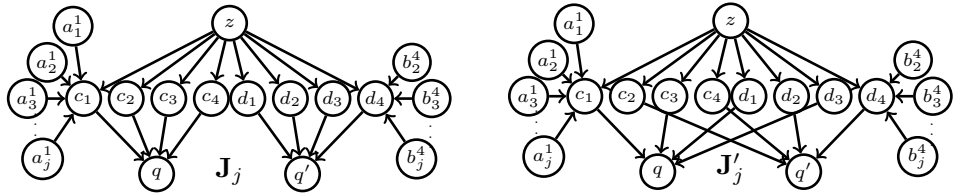

%% file: Images/IjPSPACE.tex
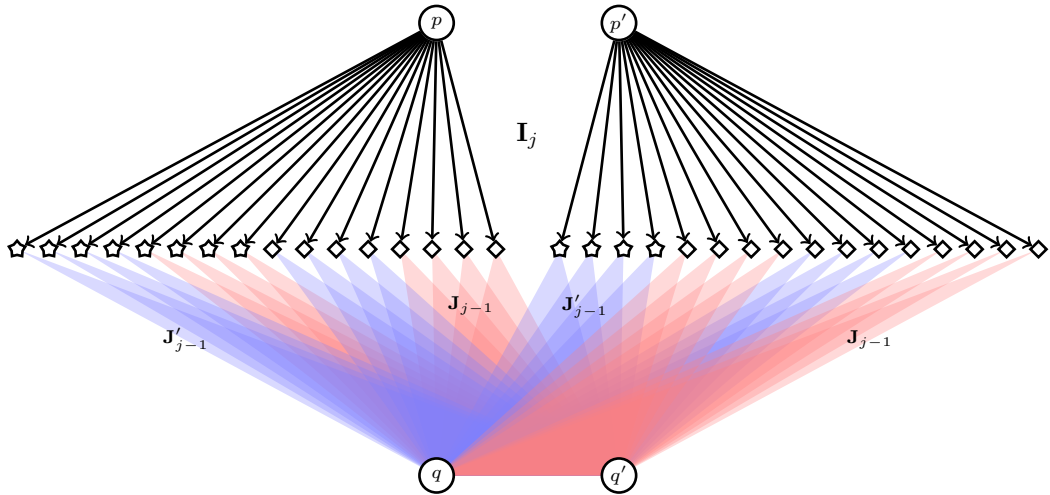
\begin{figure*}[ht]
\begin{center}
\tikzset{Blank/.style={circle,fill=none,draw=none, inner sep=1pt, minimum size=1cc, line width=1pt}}
\tikzset{Black/.style={circle,fill=white,draw=black, inner sep=1pt, minimum size=1cc, line width=1pt}}
\tikzset{BlackS/.style={star,fill=white,draw=black, inner sep=1pt, minimum size=0.5cc, line width=1pt}}
\tikzset{BlackD/.style={diamond,fill=white,draw=black, inner sep=1pt, minimum size=0.5cc, line width=1pt}}

\begin{tikzpicture}[scale=.15]

\foreach \x in {1, ..., 16}
\node [Blank] (i\x) at (-2.8*\x, 0) {};
\foreach \x in {1, ..., 16}
\node [Blank] (j\x) at (2.8*\x, 0) {};
\node [Blank] (l) at(-8,-20)  {{\scriptsize $q$}};
\node [Blank] (r) at(8,-20)  {{\scriptsize  $q'$}};

\foreach \x in {1,2,3,4,9,10,11,12}
\filldraw[draw=none, fill=red!50, opacity=0.3] (i\x.center) -- (l.center) -- (r.center) -- (i\x.center) -- cycle;
\foreach \x in {5,6,7,8,13,14,15,16}
\filldraw[draw=none, fill=blue!50, opacity=0.3] (i\x.center) -- (l.center) -- (r.center) -- (i\x.center) -- cycle;
\foreach \x in {1,2,3,4,9,10,11,12}
\filldraw[draw=none, fill=blue!50, opacity=0.3] (j\x.center) -- (l.center) -- (r.center) -- (j\x.center) -- cycle;
\foreach \x in {5,6,7,8,13,14,15,16}
\filldraw[draw=none, fill=red!50, opacity=0.3] (j\x.center) -- (l.center) -- (r.center) -- (j\x.center) -- cycle;

\node [Black] (p) at(-8,20)  {{\scriptsize $p$}};
\node [Black] (p') at(8,20)  {{\scriptsize  $p'$}};
\node [Black] (q) at(-8,-20)  {{\scriptsize $q$}};
\node [Black] (q') at(8,-20)  {{\scriptsize  $q'$}};
\foreach \x in {1, ..., 8}
\node [BlackD] (i\x) at (-2.8*\x, 0) {};
\foreach \x in {9, ..., 16}
\node [BlackS] (i\x) at (-2.8*\x, 0) {};
\foreach \x in {1, ..., 4}
\node [BlackS] (j\x) at (2.8*\x, 0) {};
\foreach \x in {5, ..., 16}
\node [BlackD] (j\x) at (2.8*\x, 0) {};

\foreach \x in {1, ..., 16}
\draw[line width=1pt,color=black,->] (p) -- (i\x);
\foreach \x in {1, ..., 16}
\draw[line width=1pt,color=black,->] (p') -- (j\x);

\node [Blank] (label) at (0, 10) {$\bI_j$};
\node [Blank] (label) at (-30, -8) {\scriptsize $\bJ'_{j-1}$};
\node [Blank] (label) at (30, -8) {\scriptsize $\bJ_{j-1}$};
\node [Blank] (label) at (-5, -5) {\scriptsize $\bJ_{j-1}$};
\node [Blank] (label) at (5, -5) {\scriptsize $\bJ'_{j-1}$};

\end{tikzpicture}
\end{center}
\caption{Gadget $\bI_j$ used in the $\PSPACE$ construction. The vertices $c_i$ for $1 \leq i \leq 12$ are shown as stars; the vertices $d_i$ for $1 \leq i \leq 20$ are shown as diamonds. Each $c_i$ is connected to $j$ independent in-neighbors, and each $d_i$ to $j - 1$ independent in-neighbors. These sets of neighbors are disjoint and not pictured. Each of these $c_i$'s and $d_i$'s is connected to $(q, q')$ by means of $\bJ_{j-1}$ or $\bJ'_{j-1}$ gadgets, depicted as \textcolor{red}{red} and \textcolor{blue}{blue} triangles respectively (with the $c_i$ or $d_i$ acting as the vertex $z$ from Fig.\ \ref{figJj}, as noted before). Note that this gadget $\bI_j$ exists for any $j > 1$.}
\label{figIjPSPACE}
\end{figure*}

%% file: Images/APSPACE.tex
\begin{figure*}
\begin{center}
\tikzset{Blank/.style={circle,fill=none,draw=none, inner sep=1pt, minimum size=1cc, line width=1pt}}
\tikzset{Black/.style={circle,fill=white,draw=black, inner sep=1pt, minimum size=1cc, line width=1pt}}
\tikzset{BlackS/.style={star,fill=white,draw=black, inner sep=1pt, minimum size=0.5cc, line width=1pt}}
\tikzset{BlackD/.style={diamond,fill=white,draw=black, inner sep=1pt, minimum size=0.5cc, line width=1pt}}
\tikzset{Big/.style={circle,fill=white,draw=black, inner sep=1pt, minimum size=5cc, line width=1pt}}

\begin{tikzpicture}[scale=.15]

\foreach \x in {1, ..., 16}
\node [Blank] (i\x) at (-2.8*\x, 0) {};
\foreach \x in {1, ..., 16}
\node [Blank] (j\x) at (2.8*\x, 0) {};
\node [Blank] (l) at(-8,-20)  {{\scriptsize $q$}};
\node [Blank] (r) at(8,-20)  {{\scriptsize  $q'$}};

\foreach \x in {1,2,3,4,9,10,11,12}
\filldraw[draw=none, fill=red!50, opacity=0.3] (i\x.center) -- (l.center) -- (r.center) -- (i\x.center) -- cycle;
\foreach \x in {5,6,7,8,13,14,15,16}
\filldraw[draw=none, fill=blue!50, opacity=0.3] (i\x.center) -- (l.center) -- (r.center) -- (i\x.center) -- cycle;
\foreach \x in {1,2,3,4,9,10,11,12}
\filldraw[draw=none, fill=blue!50, opacity=0.3] (j\x.center) -- (l.center) -- (r.center) -- (j\x.center) -- cycle;
\foreach \x in {5,6,7,8,13,14,15,16}
\filldraw[draw=none, fill=red!50, opacity=0.3] (j\x.center) -- (l.center) -- (r.center) -- (j\x.center) -- cycle;

\node [Black] (p) at(-8,20)  [label=left:{\scriptsize  $p_k$}]{};
\node [Black] (p') at(8,20)  [label=right:{\scriptsize  $p'_k$}]{};
\node [Black] (q) at(-8,-20)  {{\scriptsize $q$}};
\node [Black] (q') at(8,-20)  {{\scriptsize  $q'$}};
\foreach \x in {1, ..., 8}
\node [BlackD] (i\x) at (-2.8*\x, 0) {};
\foreach \x in {9, ..., 16}
\node [BlackS] (i\x) at (-2.8*\x, 0) {};
\foreach \x in {1, ..., 4}
\node [BlackS] (j\x) at (2.8*\x, 0) {};
\foreach \x in {5, ..., 16}
\node [BlackD] (j\x) at (2.8*\x, 0) {};

\foreach \x in {1, ..., 16}
\draw[line width=1pt,color=black,->] (p) -- (i\x);
\foreach \x in {1, ..., 16}
\draw[line width=1pt,color=black,->] (p') -- (j\x);

\node [Blank] (label) at (0, 10) {$\bI_{2k + m - t}$};
\node [Blank] (label) at (-30, -8) {\scriptsize $\bJ'_{2k + m - t - 1}$};
\node [Blank] (label) at (30, -8) {\scriptsize $\bJ_{2k + m - t - 1}$};
\node [Blank] (label) at (-5, -5) {\scriptsize $\bJ_{2k + m - t - 1}$};
\node [Blank] (label) at (5, -5) {\scriptsize $\bJ'_{2k + m - t - 1}$};

\def \off {40};

\foreach \x in {1, ..., 16}
\node [Blank] (i\x) at (-2.8*\x, 0-\off) {};
\foreach \x in {1, ..., 16}
\node [Blank] (j\x) at (2.8*\x, 0-\off) {};
\node [Blank] (l) at(-8,-20-\off)  {{\scriptsize $q$}};
\node [Blank] (r) at(8,-20-\off)  {{\scriptsize  $q'$}};

\foreach \x in {1,2,3,4,9,10,11,12}
\filldraw[draw=none, fill=red!50, opacity=0.3] (i\x.center) -- (l.center) -- (r.center) -- (i\x.center) -- cycle;
\foreach \x in {5,6,7,8,13,14,15,16}
\filldraw[draw=none, fill=blue!50, opacity=0.3] (i\x.center) -- (l.center) -- (r.center) -- (i\x.center) -- cycle;
\foreach \x in {1,2,3,4,9,10,11,12}
\filldraw[draw=none, fill=blue!50, opacity=0.3] (j\x.center) -- (l.center) -- (r.center) -- (j\x.center) -- cycle;
\foreach \x in {5,6,7,8,13,14,15,16}
\filldraw[draw=none, fill=red!50, opacity=0.3] (j\x.center) -- (l.center) -- (r.center) -- (j\x.center) -- cycle;

\node [Black] (p) at(-8,20-\off)  [label=left:{\scriptsize  $q_k = p_{k-1}$}]{};
\node [Black] (p') at(8,20-\off)  [label=right:{\scriptsize  $q'_k = p'_{k-1}$}]{};
\node [Black] (q) at(-8,-20-\off)  [label=left:{\scriptsize  $q_{k-1} = p_{k-2}$}]{};
\node [Black] (q') at(8,-20-\off)  [label=right:{\scriptsize  $q'_{k-1} = p'_{k-2}$}]{};
\foreach \x in {1, ..., 8}
\node [BlackD] (i\x) at (-2.8*\x, 0-\off) {};
\foreach \x in {9, ..., 16}
\node [BlackS] (i\x) at (-2.8*\x, 0-\off) {};
\foreach \x in {1, ..., 4}
\node [BlackS] (j\x) at (2.8*\x, 0-\off) {};
\foreach \x in {5, ..., 16}
\node [BlackD] (j\x) at (2.8*\x, 0-\off) {};

\foreach \x in {1, ..., 16}
\draw[line width=1pt,color=black,->] (p) -- (i\x);
\foreach \x in {1, ..., 16}
\draw[line width=1pt,color=black,->] (p') -- (j\x);

\node [Blank] (label) at (0, 10-\off) {$\bI_{2k + m - t - 2}$};
\node [Blank] (label) at (-30, -8-\off) {\scriptsize $\bJ'_{2k + m - t - 3}$};
\node [Blank] (label) at (30, -8-\off) {\scriptsize $\bJ_{2k + m - t - 3}$};
\node [Blank] (label) at (-5, -5-\off) {\scriptsize $\bJ_{2k + m - t - 3}$};
\node [Blank] (label) at (5, -5-\off) {\scriptsize $\bJ'_{2k + m - t - 3}$};

\def \off {90};

\foreach \x in {1, ..., 16}
\node [Blank] (i\x) at (-2.8*\x, 0-\off) {};
\foreach \x in {1, ..., 16}
\node [Blank] (j\x) at (2.8*\x, 0-\off) {};
\node [Blank] (l) at(-8,-20-\off)  {{\scriptsize $q$}};
\node [Blank] (r) at(8,-20-\off)  {{\scriptsize  $q'$}};

\foreach \x in {1,2,3,4,9,10,11,12}
\filldraw[draw=none, fill=red!50, opacity=0.3] (i\x.center) -- (l.center) -- (r.center) -- (i\x.center) -- cycle;
\foreach \x in {5,6,7,8,13,14,15,16}
\filldraw[draw=none, fill=blue!50, opacity=0.3] (i\x.center) -- (l.center) -- (r.center) -- (i\x.center) -- cycle;
\foreach \x in {1,2,3,4,9,10,11,12}
\filldraw[draw=none, fill=blue!50, opacity=0.3] (j\x.center) -- (l.center) -- (r.center) -- (j\x.center) -- cycle;
\foreach \x in {5,6,7,8,13,14,15,16}
\filldraw[draw=none, fill=red!50, opacity=0.3] (j\x.center) -- (l.center) -- (r.center) -- (j\x.center) -- cycle;

\node [Black] (p) at(-8,20-\off)  [label=left:{\scriptsize  $q_2 = p_1$}]{};
\node [Black] (p') at(8,20-\off)  [label=right:{\scriptsize  $q'_2 = p'_1$}]{};
\node [Big] (q) at(-8,-20-\off)  {};
\node [Big] (q') at(8,-20-\off)  {};
\foreach \x in {1, ..., 8}
\node [BlackD] (i\x) at (-2.8*\x, 0-\off) {};
\foreach \x in {9, ..., 16}
\node [BlackS] (i\x) at (-2.8*\x, 0-\off) {};
\foreach \x in {1, ..., 4}
\node [BlackS] (j\x) at (2.8*\x, 0-\off) {};
\foreach \x in {5, ..., 16}
\node [BlackD] (j\x) at (2.8*\x, 0-\off) {};

\node [Black] (C1l) at(-8-4,-20-\off)  {{\scriptsize $v_{C_1}$}};
\node [Black] (Cml) at(-8+4,-20-\off)  {{\scriptsize $v_{C_m}$}};

\node [Black] (C1r) at(8-3.5,-20+3-\off)  {{\scriptsize $v'_{C_1}$}};
\node [Black] (Cmr) at(8+3.5,-20+3-\off)  {{\scriptsize $v'_{C_m}$}};
\node [Black] (null1) at(8-3.5,-20-3-\off)  {{\scriptsize $\mathsf{null}_1$}};
\node [Black] (nullm) at(8+3.5,-20-3-\off)  {{\scriptsize $\mathsf{null}_m$}};

\node [Blank] (dots1) at(-8,-20-\off)  {{\scriptsize $\cdots$}};
\node [Blank] (dots2) at(8,-20-3-\off)  {{\scriptsize $\cdots$}};
\node [Blank] (dots3) at(8,-20+3-\off)  {{\scriptsize $\cdots$}};
\node [Blank] (dots4) at(-8,25-\off)  {$\vdots$};
\node [Blank] (dots5) at(8,25-\off)  {$\vdots$};

\foreach \x in {1, ..., 16}
\draw[line width=1pt,color=black,->] (p) -- (i\x);
\foreach \x in {1, ..., 16}
\draw[line width=1pt,color=black,->] (p') -- (j\x);

\node [Blank] (label) at (0, 10-\off) {$\bI_{2 + m - t}$};
\node [Blank] (label) at (-30, -8-\off) {\scriptsize $\bJ'_{1 + m - t}$};
\node [Blank] (label) at (30, -8-\off) {\scriptsize $\bJ_{1 + m - t}$};
\node [Blank] (label) at (-5, -5-\off) {\scriptsize $\bJ_{1 + m - t}$};
\node [Blank] (label) at (5, -5-\off) {\scriptsize $\bJ'_{1 + m - t}$};

%\draw[line width=1pt,color=black,->]  (C1l) to[out=160,in=-150] (-8,-50);
\draw[line width=1pt,color=black,->]  (C1l) to[out=150,in=-150] (-12,-100);
\draw[line width=1pt,color=black,->]  (C1l) to[out=150,in=-150] (i6);

\end{tikzpicture}
\end{center}
\caption{The structure $\bA$, built using gadgets $\bI_{2+m-t}, \bI_{4+m-t}, \ldots, \bI_{2k + m - t}$ stacked together as shown. Only a few outgoing edges from the ``clause'' vertex $v_{C_1}$ are shown, and only to the gadget $\bI_2$. In general, there will be edges from the clause vertices $v_{C_i}$ and $v'_{C_i}$ to every gadget that contains vertices which have ``informal labels'' corresponding to literals that appear in the clause $C_i$ of the quantifier-free part of the formula $\varphi$.}
\label{figAPSPACE}
\end{figure*}
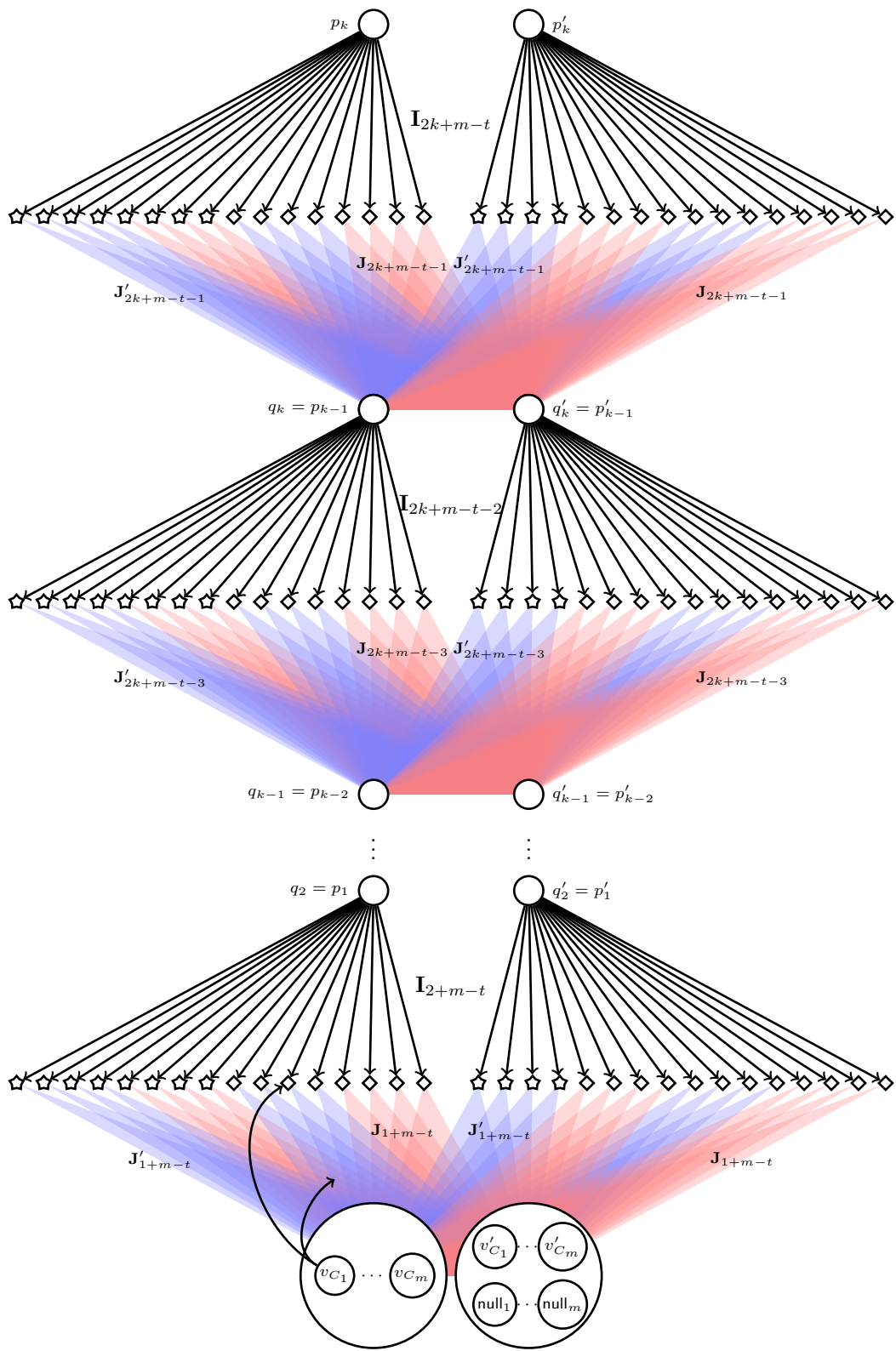

%% file: 05discussion.tex
\section{Concluding Remarks}
\label{sec:discussion}

\commentout{
\begin{table}[ht] 
\begin{center}
\begin{tabular}{|c|c|c|c|c|}\hline 
\textbf{Game} & \textbf{Capture} & \textbf{Upper} & \textbf{Lower} & \textbf{Reference} \\ \hline
  {\EF} & Quantifier Rank & $\PSPACE$ & $\PSPACE$-hard & Pezzoli \cite{Pezzoli98} \\ \hline
  ? & ? & ? & ? & ? \\ \hline
  ? & ? & ? & ? & ? \\ \hline
  ? & ? & ? & ? & ? \\ \hline
  Multi-Structural & Quantifier Number & $\NEXP$ & $\PSPACE$-hard & \textcolor{red}{this paper} \\ \hline
\end{tabular}
\caption{Landscape of two-player combinatorial games capturing first-order logical properties.} 
\label{landscape-table}
\end{center}
\end{table}
%end-commentout
}

In this paper, we investigated the computational complexity of determining whether Spoiler wins a given instance of the {\ms} game. When the number of moves is fixed, this problem is solvable in logarithmic space. When the number of moves is part of the input, we showed that this problem is $\PSPACE$-hard and also solvable in non-deterministic exponential time. Table \ref{table:summary} summarizes the state of affairs concerning the complexity of determining the winner in the multi-structural game and the other three games discussed in the introduction.

\begin{table}[h!]
\begin{center}
\begin{tabular}{|c|c|c|c|} 
\hline
{\bf Game}  & {\bf Problem} & {\bf Complexity} & {\bf Reference} \\
\hline

\multirow{2}{9em}{\EF} & $\winef$ & $\PSPACE$-complete & \cite{Pezzoli98}  \\ 
& $\winef_m$, $m \geq 2$   & in $\LOGSPACE$  & \cite{Pezzoli98} \\ 
\hline

\multirow{2}{9em}{Pebble} & $\winpb$ & \textcolor{red}{open} &   \\ 
& $\winpb_k$, $k \geq 2$   & $\P$-complete  & \cite{GroheLk} \\ 
\hline

\multirow{2}{9em}{Existential Pebble} & $\winepb$ & $\EXP$-complete &  \cite{DBLP:conf/csl/KolaitisP03} \\ 
& $\winepb_k$, $k \geq 2$ & $\P$-complete  & \cite{DBLP:conf/csl/KolaitisP03} \\ 
\hline

\multirow{2}{9em}{Multi-Structural} & $\winms$ & $\PSPACE$-hard,   in $\NEXP$ & \textcolor{blue}{this work} \\ 
& $\winms_m$, $m \geq 2$   & in $\LOGSPACE$  & \textcolor{blue}{this work} \\
\hline
\end{tabular}
\caption{Summary of complexity results for determining the winner of combinatorial games.}
\label{table:summary}
\end{center}
\end{table}

The next step in this investigation is to close the gap between $\PSPACE$-hardness and membership in $\NEXP$ for $\winms$. We conjecture that $\winms$ is a $\NEXP$-complete problem, and hope that the work reported here will provide the impetus for further research on this conjecture.

We conclude by pointing out that all hardness and completeness results in Table \ref{table:summary} hold under the assumption that the input structures are \emph{unordered}, i.e., it is not assumed that one of the relations in each of the given structures is a total order on the universe of the structure. As is well known, many results in descriptive complexity (e.g., the Immerman-Vardi theorem \cite{DBLP:journals/iandc/Immerman86,DBLP:conf/stoc/Vardi82} that least fixed-point logic captures
%\jon{Should be `captures'. Perhaps also specify that it is First Order least fixed-point logic that captures P-time.}
polynomial time) depend crucially on the assumption that the structures are ordered. 
However, the reductions used to derive the complexity results in Table \ref{table:summary} do not work if the structures at hand are ordered.
Thus, it would be interesting to investigate the complexity of determining the winner in each of the games discussed here when the input structures are ordered.